\documentclass[twoside,leqno,twocolumn]{article}

\usepackage[letterpaper]{geometry}

\usepackage{ltexpprt}
\usepackage{hyperref}

\usepackage{graphicx}
\usepackage{xcolor}
\usepackage{booktabs}
\usepackage{amsmath}
\usepackage{amssymb}

\newif\ifappendix
\appendixtrue

\newif\ifenablecomments

\enablecommentsfalse

\newcommand\ifcomments[1]{{
    \ifenablecomments
        #1
    \else
    \fi}}

\newcommand{\pcmax}{\texorpdfstring{$P||C_{\mathrm{max}}$}{P||Cmax}}
\newcommand{\cmax}{\texorpdfstring{C_{\mathrm{max}}}{Cmax}}
\newcommand\todo[1]{\ifcomments{\textcolor{orange}{#1}}}
\newcommand{\pcmaxinstance}{$(W, m)$} %{$(J, W, n, m)$}
\newcommand{\pcmaxdecisioninstance}{$(W, m, U)$} %{$(J, W, n, m, U)$}
\newcommand{\N}{\mathbb{N}}
\newcommand{\bnb}{BnB}
\newcommand{\decisionlvl}{\ell}
\newcommand{\decisionlvlplus}{i}
\newcommand{\wlogen}{\mbox{w.l.o.g.}}
\newcommand{\funcleft}{\textsl{left}}
\newcommand{\funcright}{\textsl{right}}
\DeclareMathOperator{\modnospace}{mod}
\newcommand{\psfrage}[1]{\ifcomments{{\color{blue}{\sf[PS: #1]}}}} % Peter Sanders
\newcommand{\algalgorithm}{CDSM}
\newcommand{\algtrue}{\textsc{true}}
\newcommand{\algfalse}{\textsc{false}}
\newcommand{\algbestsolution}{$S$}
\newcommand{\algstatetable}{ST}
\newcommand{\algirrelevanceindex}{$\rho$}
\newcommand{\algret}{\textsc{RET}}

\newtheorem{pruning rule}{Pruning Rule}
\usepackage{thm-restate}

\usepackage{biblatex}
\begin{document}

\newcommand\relatedversion{}

\title{\Large Engineering Optimal Parallel Task Scheduling} %\relatedversion}
\author{Matthew Akram\thanks{Karlsruhe Institute of Technology, Karlsruhe, Germany. mazfh85246@gmail.com, nikolai.maas@kit.edu, sanders@kit.edu, dominik.schreiber@kit.edu}
\and Nikolai Maas\footnotemark[1]
\and Peter Sanders\footnotemark[1]
\and Dominik Schreiber\footnotemark[1]}
%Corey Gray\thanks{Society for Industrial and Applied Mathematics.}
%\and Tricia Manning\thanks{Society for Industrial and Applied Mathematics.}}

\date{}

\maketitle

% Copyright Statement
% When submitting your final paper to a SIAM proceedings, it is requested that you include
% the appropriate copyright in the footer of the paper.  The copyright added should be
% consistent with the copyright selected on the copyright form submitted with the paper.
% Please note that "20XX" should be changed to the year of the meeting.

% Default Copyright Statement
\ifappendix\else
\fancyfoot[R]{\scriptsize{Copyright \textcopyright\ 2025 by SIAM\\
Unauthorized reproduction of this article is prohibited}}
\fi

% Depending on which copyright you agree to when you sign the copyright form, the copyright
% can be changed to one of the following after commenting out the default copyright statement
% above.

%\fancyfoot[R]{\scriptsize{Copyright \textcopyright\ 20XX\\
%Copyright for this paper is retained by authors}}

%\fancyfoot[R]{\scriptsize{Copyright \textcopyright\ 20XX\\
%Copyright retained by principal author's organization}}

%\pagenumbering{arabic}
%\setcounter{page}{1}%Leave this line commented out.

\begin{abstract} \small\baselineskip=9pt 
The NP-hard scheduling problem \pcmax{} encompasses a set of tasks with known execution time which must be mapped to a set of identical machines such that the overall completion time is minimized.
In this work, we improve existing techniques for optimal \pcmax{} scheduling with a combination of new theoretical insights and careful practical engineering.
Most importantly, we derive techniques to prune vast portions of the search space of branch-and-bound (\bnb{}) approaches
%We also 
and propose improved upper and lower bounding techniques. % which can be combined with any approach to \pcmax{}.
Moreover, we present new benchmarks for \pcmax{}, based on diverse applications, which can shed light on aspects that prior synthetic instances fail to capture.
In extensive evaluations, we observe that our pruning reduces the number of explored nodes by 90$\times$ and running times by 12$\times$.
Compared to a state-of-the-art ILP-based approach, our approach is preferable for short running time limits and for instances with large makespans.
%substantial improvements, often of exponential magnitude, both for \bnb{}- as well as ILP-based scheduling.
\end{abstract}

\vspace{.2cm}
Artifact: \href{https://zenodo.org/records/13835637}{https://zenodo.org/records/13835637}

\section{Introduction}

One of the most fundamental yet still challenging load balancing problems in computer science is the problem of \emph{unconstrained, non-preemptive task scheduling} on \emph{identical parallel machines}.
This problem is often denoted by its Graham notation \pcmax{} \cite{graham1979optimization},
where $P$ represents identical machines running in parallel.
The jobs to process are given as a set of integers, each indicating the time (or \emph{work}) a job takes.
Our objective is to minimize $\cmax$, the maximum makespan (completion time) of any machine.
\pcmax{} is a strongly NP-hard problem~\cite{GareyJohnson:78:Strong-NP-completeness} and, as such, demands either inexact solving or exponential complexity~\cite{berndt2022load}
%\todo{NM: a somewhat odd statement considering that our approach is neither. How about ``either approximate or exponential approaches''? Leaving the distinction between simple heuristics and approximation schemes to the related work}
in order to conquer non-trivial instances (unless $P=NP$).
We focus on \emph{branch-and-bound} (\bnb) schemes~\cite{dellamico_martello_1995,walter2020characterization}, which are able to find \textit{optimal solutions} by heuristically pruning the search space.
%\todo{NM: I don't know whether an algorithm that is guaranteed to find the optimal solution can be called a ``heuristic''.
%DS: I think so. E.g. in graph search admissible heuristics lead to optimal solutions.
%In that sense heuristics are used to prioritize and safely prune, which doesn't need to imply inaccuracies/suboptimality.
%NM: Yeah but if we are precise, isn't the ``heuristic'' just the decision making (i.e., pruning) algorithm and \emph{not} the whole algorithm? Specifically, an admissible heuristic in graph search does not give the exact value but a one-sided approximation.
%In any case, I would suggest somesting like ``utilize algorithms/methods which are designed to (eventually) find an optimal solution while using heuristics to prune the search space''
%}
%Among such methods, we especially consider 
%and approaches based on Integer Linear Programming (ILP)~\cite{dimm,8248744}.%; and translations to propositional logic (SAT), which have previously been successful in the context of related scheduling problems~\cite{crawford1994experimental,koshimura2010solving}.
% \todo{NM: probably remove last part of sentence if we don't focus on SAT anymore}

In this work, we advance existing techniques for optimal \pcmax{} scheduling.
%\todo{NM: I don't like this sentence, though I'm not entirely sure why. Perhaps because it seems out of place.
%How about: In this work, we describe and evaluate techniques to improve the practical performance of these methods (, using the {Algorithm Engineering}~\cite{sanders2009algorithm} methodology).
%Just removing it also seems okay.
%DS: OK? I do think a little glue sentence makes sense here.
%}
With the original intention of exploring efficient reductions of \pcmax{} to propositional satisfiability (SAT), we have investigated ways to constrain the search space of \pcmax{} \emph{decision instances}.
%\todo{NM: Maybe this sentence and our overall storyline for the motivation should be ``using insights from SAT for \pcmax{}'' instead of ``exploiting SAT for \pcmax{}''?
%DS: It wasn't intended to sound like that's the main topic of the paper but more like that's the direction we were originally coming from.
%How about: With the original intention of exploring efficient reductions of \pcmax{} to propositional satisfiability (SAT), we have investigated \ldots
%NM: Yes that's better
%}
The objective of a decision instance is to decide for a particular $U$ whether a solution with $\cmax \leq U$ exists.
Focusing on this decision problem yields a series of interesting and powerful insights which
%prove useful for \pcmax{} scheduling approaches beyond SAT.
prove especially useful for \bnb{} based scheduling approaches.
%many different scheduling algorithms and, in turn, advance the state of the art in \pcmax{} scheduling.

Our contributions are the following.
First and foremost, we introduce new pruning rules based on the \pcmax{} decision problem which can prune vast portions of an instance's effective search space, especially for \bnb{} (Section~\ref{sec:search-space-pruning}).
Secondly, we improve existing techniques to compute lower and upper bounds for $\cmax{}$ (Section~\ref{sec:bounds}).
We then present an efficient implementation of our techniques within a \bnb{} scheme (Section~\ref{sec:implementation}).
Furthermore, we introduce a new, large benchmark set crafted from realistic and diverse application data to complement prior, purely synthetic benchmark sets. %considered in most prior works.
We evaluate our contributions in thorough experiments (Section~\ref{sec:evaluation}) %to state-of-the-art \pcmax{} scheduling approaches 
and observe drastic improvements.
In particular, our \bnb{} pruning techniques reduce explored nodes by two orders of magnitude, which results in substantial speedups and allows to solve 64\% more instances. % than a base implementation. %, and roughly 15\% more than the HJ~\cite{lifting_2} heuristic from the literature.
%Even for huge input sizes, our techniques solve the vast majority of application instances within seconds.
Compared to a state-of-the-art approach based on Integer Linear Programming (ILP)~\cite{8248744}, our approach solves fewer instances but is significantly faster and tends to scale better to large makespans.
%\todo{DS: That last sentence is problematic I think.}
% =================================================

Some preliminary results of our work were briefly announced in
%Ref.\  [\textit{todo: not yet published}]\nolinebreak
\cite{10.1145/3626183.3660268}\nolinebreak
---only covering rules R\ref{pruning_rule:inter} and R\ref{pruning_rule:fur}.

\section{Preliminaries}\label{sec:preliminaries}
% definitions, related work
% \todo{MA:Remove this sentence for space}
% In the following, we provide some relevant preliminaries for our work.

%\subsection{Problem Definition}

%\psfrage{discuss: is this a standard notation? In my book I try a more compact notation but I agree that it has its idiosyncrasies.}\todo{NM: Yes it is standard notation (so-called Graham notation, en.wikipedia.org/wiki/Optimal\_job\_scheduling)}\psfrage{This is a misunderstanding, I am not alluding to the triple notation but to $P = \{p_1, \ldots, p_m\}$, $J = \{j_1, \ldots, j_n\}$, in my book I basically have $P = \{1, \ldots, m\}$, $J = \{1, \ldots, n\}$ and write a lot of things as vectors.}
A \pcmax{} \emph{optimization instance} $(W, m)$ is defined by $n$ durations $W = \{w_1, \ldots, w_n\}$ of $n$ corresponding jobs $J = \{j_1, \ldots, j_n\}$ and by the number $m$ of identical processors $P = \{p_1, \ldots, p_m\}$.
An instance $(W, m)$ asks for an \emph{assignment} $A = \{ a_1, \ldots, a_n \}$ of jobs to machines ($1 \leq a_i \leq m$) such that the \emph{maximum completion time} $\cmax := \max_{i}\{ \sum_{k\,|\,a_k = i} w_k \}$ is minimized.
By contrast, a \emph{decision instance} $(W, m, U)$ additionally imposes an \emph{upper bound} $U$ for $\cmax$ and poses the question whether a feasible solution exists.
Throughout this paper we assume that the jobs are sorted by duration in decreasing order (i.e., $w_1 \geq w_2 \geq \ldots \geq w_n$).

%\subsection{Branch-and-bound}

A \emph{branch-and-bound} (\bnb{}) algorithm for \pcmax{} is a tree-like search where we extend an initially empty \emph{partial assignment} $A$ of jobs to processors until $|A| = n$.
At each \emph{decision level} $\decisionlvl$, where $|A|=\decisionlvl$,
%\psfrage{say this more concretely? "At decision level $\ell$ we assign job $j_\ell$ to a processor"?}
%we identify a set of decisions (``branches'') which extend $A$ by assigning job $j_{\decisionlvl+1}$ to a specific processor.
we assign job $j_{\decisionlvl+1}$ to a processor.
Each possible assignment constitutes a decision (``branch'') of the algorithm.
%\emph{Pruning rules} allow to ignore some of these decisions because they are dominated by other decisions or result in an infeasible problem.
In addition, we maintain admissible bounds on $\cmax$ during the search, which allows us to exclude decisions which inevitably lead to sub-optimal solutions (``bound'').
We recursively search the remaining decisions in a heuristically chosen order.
For each partial assignment $A$, we define the \emph{assigned workload} of processor $p_x$ as $C_x^A := \sum_{i\,|\,a_i=x} w_i$.
The \emph{least loaded} processor is the one with smallest $C_{x}^A$.

\subsection{Related Work.}

\pcmax{} has been extensively studied as it has a simple formulation
% with many applications
% --> NM: I'm afraid I don't believe there are many (direct) application. Even if there is a notable number, we should include some references for the claim 
and is yet difficult to solve, with many interesting aspects % and approaches
both in theory and practice.
%\pcmax{} has been the subject of numerous publications due to both its difficulty and \psfrage{add sth about wide range of applications? Or "\pcmax{} has been intensively studied as it has at the same time a simple formulation with many applications yet is also difficult to solve with many interesting aspects and approaches both in theory and practice."}simplicity.
For a thorough summary of works on \pcmax{}, we refer to the excellent related work sections by Lawrinenko et al.~\cite{dissertation_lawrinenko,walter2020characterization} and Mrad et al.~\cite{8248744}, covering most works up to 2017.
More recent research has mostly focused on generalized problem variants with additional constraints.
Since the \pcmax{} \emph{decision} problem is equivalent to the widely-studied bin packing problem (BPP),
exact algorithms for the BPP can be of interest for \pcmax{} scheduling.
While a further comparison is out of scope for this work, we refer to Delorme et al.~\cite{bp_survey_2016} for an overview of such techniques.

%Finding \emph{some} (non-optimal) schedule is always feasible for \pcmax{}.
%For example, \mbox{Graham}'s \emph{longest-processing-time-first (LPT)} strategy,
A simple and fast \pcmax{} approximation is provided by \mbox{Graham}'s \emph{longest-processing-time-first} (LPT) strategy,
which assigns jobs in decreasing order of duration to the currently least loaded processor
%, is computationally dominated by sorting
and has an approximation ratio of 4/3~\cite{graham1969bounds}.
%\psfrage{if we have enough space, it might be good to mention that better fast algorithms like best-fit are also based on solving a sequence of decision problems.}%\todo{Mention the duality of \pcmax{} decision, and BPP decision}.
%old formulation:
%There are also \textsc{ptas} and \textsc{eptas} for \pcmax{}, which we consider to be of mostly theoretical value~(cf.\ \cite{berndt2022load}) and do not cover any further. % in this work.
%new formulation:
There are also %\textsc{ptas} and \textsc{eptas}
\mbox{EPTAS} for \pcmax{},
which provide arbitrarily good approximations in theory.
However, achieving approximations that substantially outperform simpler algorithms is rather difficult in practice~\cite{berndt2022load}.

%Since \pcmax{} is an optimization problem, we can
An important technique is the identification of \emph{upper and lower bounds} on the optimal makespan of a given \pcmax{} instance.
%These bounds tell us that the optimal makespan of the given instance is within a given range which we can restrict our search to. %only within the given bounds.
Restricting the range of admissible makespans can greatly benefit exact solvers, especially if %at least
one of the given bounds is tight.
%We revisit some of the following literature on upper and lower bounding techniques in greater detail in Section~\ref{sec:upper_and_lower_bounds_lit}.
%The most advanced lower bounding techniques are presented by 
Haouari et al.~\cite{lifting_1, lifting_2} present
%the currently most advanced
state-of-the-art lower bounding techniques, as well as two lifting procedures to further improve lower bounding algorithms.
They also propose the \emph{multi-start subset sum} upper bounding technique, which is essentially a local search scheme with restarts.
Dell’Amico et al.~\cite{dellamico2008heuristic} introduce a scatter-search based bounding technique that gives their exact solving scheme excellent performance.
%The literature on upper bounds for \pcmax{} is vast~\cite{dissertation_lawrinenko},
%in parts due to the fact that \emph{any} \pcmax{} scheduling approach constitutes an upper bounding technique.
Since \emph{any} \pcmax{} scheduling approach constitutes an upper bounding technique,
%the overall 
literature on upper bounds is vast~\cite{dissertation_lawrinenko}.
%For an overview, we refer to Lawrinenko's dissertation.

We now turn to exact solving procedures.
% is a little more sparse. %, with the historically first approaches using exponential-time dynamic programming~\cite{f6234714-5a67-370d-9a9e-dcc647939b39}
Historically one of the first approaches, Rothkopf~\cite{rothkopf1966scheduling} presented in 1966 a dynamic programming approach in $O(n\cdot U^m)$.
%This algorithm is slow, and can only solve very small instances.
Mokottof~\cite{mokotoff2004exact} proposes a cutting plane algorithm where valid inequalities are identified and added to an LP encoding until the respective solution is integer.
%The first \bnb{} algorithm for \pcmax{} was introduced by Dell'Amico and Martello~\cite{dellamico_martello_1995} who introduce novel pruning rules to decrease the size of the search space of the \bnb{} algorithm.
Dell'Amico and Martello~\cite{dellamico_martello_1995} introduced the first \bnb{} algorithm for \pcmax{} as well as pruning rules to decrease the size of the search space.
Haouari and Jemmali~\cite{lifting_2} use a different branch order to construct a novel \bnb{} algorithm %which far exceeds the performance of
which compares favorably to the algorithm by Dell'Amico and Martello.
%During experimentation, they detect a set of difficult \pcmax{} instances, with $n/m = 2.5$,
%and argue that future work should focus on such instances, since instances with larger $n/m$ can usually be solved to optimality using simple bounding heuristics.
In addition, they argue that future work should focus on instances with $n/m$ close to 2.5,
since most other instances proved easy to solve in their experiments.
%Roughly at the same time,
%Dell'Amico et al.~\cite{dellamico2008heuristic} present a combination of a scatter search algorithm, followed by an ILP translation of the \pcmax{} decision problem, and a dedicated \bnb{} algorithm to solve the ILP translation with excellent performance, even on very large instances.
Dell'Amico et al.~\cite{dellamico2008heuristic} present a hybrid approach with excellent performance even on large instances.
It consists of a scatter search algorithm, followed by an ILP translation and a dedicated \bnb{} scheme to solve the latter.\footnote{Unfortunately, %no implementation is publicly available and
a full reimplementation of this approach was out of scope for the work at hand.}
%Unfortunately, no implementation is publicly available.
Lawrinenko~\cite{dissertation_lawrinenko, walter2020characterization} studies the structure of solutions of the \pcmax{} problem to develop new pruning rules for a \bnb{} algorithm.
The algorithm improves performance on the instances proposed by Haouari and Jemmali,
%The algorithm presented caters especially to the instances proposed by Haouari and Jemmali, and has improved performance on those instances, but greatly hinders performance
but incurs notable overhead
on other instance sets.
Finally, Mrad et al.~\cite{8248744} present a pseudo-polynomial ILP encoding of the \pcmax{} problem with great performance on the instances proposed by Haouari and Jemmali.

\section{Pruning the Search Space of \pcmax{}}\label{sec:search-space-pruning}
% TODO: section introduction

%All optimal approaches to \pcmax{} scheduling have in common that they 

%Carefully exploring %explore the instance's \emph{search space}, i.e., 
%the space of possible assignments of jobs to processors is not a unique feature of SAT-based approaches but in fact common to all optimal approaches to \pcmax{} scheduling.
The performance of
exact solvers
%all optimal approaches to \pcmax{} scheduling
heavily depends on how they explore the space of possible assignments of jobs to processors.
In the following, we investigate effective methods to prune the search space of a \pcmax{} instance.
We begin with \emph{pruning rules}---also known as \emph{dominance criteria} or \emph{symmetry breaking}---that restrict the structure of a solution
while preserving feasibility of the instance.
\todo{NM: one review comment said the distinction between optimization and decision problem is missing here, I tried to address this with the following sentence}
This includes rules which preserve at least one optimal assignment and are thus directly applicable to the \pcmax{} \emph{optimization} problem,
as well as rules which only preserve a feasible solution for the \pcmax{} \emph{decision} problem.
%while preserving at least one feasible assignment.
%We begin with discussing specific rules which restrict the structure of a solution while preserving at least one optimal assignment.
%We refer to these rules as \emph{pruning rules}---other common names include \emph{dominance criteria} or \emph{symmetry breaking}.
%In this section we consider pruning rules for a branch-and-bound algorithm on the \pcmax{} decision problem.
%We begin by discussing pruning rules from literature and then introduce a number of new pruning rules.
%We also discuss their efficient implementation.
We then show how to detect dead-ends in the search space
by transferring \emph{clause learning} techniques known from SAT solving.
%We then transfer \emph{clause learning} known from SAT solving to \bnb{}-based \pcmax{} solving %and obtain a powerful means 
%to detect dead-ends in the search space.

\subsection{Prior Pruning Rules.} %from the Literature}
% we might want to move this to related work

Dell'Amico and Martello~\cite{dellamico_martello_1995} present a number of \bnb{} pruning rules
for the \pcmax{} \emph{optimization} problem.
%Note that any pruning rule for the \pcmax{} optimization problem is also a valid pruning rule for the \pcmax{} decision problem.
For each rule we consider a partial assignment $A$ at decision level $\decisionlvl$.
Remember that we always assume the jobs are sorted by duration (i.e., $w_1 \geq \ldots \geq w_n$).

\begin{pruning rule}\label{pruning_rule:equiv_old}
 If there are multiple processors $P := \{p_x, p_y, \ldots\}$ with identical loads $C_x^A = C_y^A = \ldots$,
 then only \emph{one} processor in $P$ has to be considered when assigning $j_{\decisionlvl+1}$.
\end{pruning rule}

\begin{pruning rule}\label{pruning_rule:interchangeable_jobs}
    If $w_\decisionlvl = w_{\decisionlvl+1}$
    %\psfrage{This assumes the jobs are assigned in sorted order? Is this already explained anywhere?}
    , then only decisions with $a_{\decisionlvl+1} \le a_\decisionlvl$ need to be considered.
\end{pruning rule}

\begin{pruning rule}\label{pruning_rule:last_levels}
    If three jobs remain unassigned, only two options need to be considered:
    \emph{(1)} Assign each of the jobs to the least loaded processor respectively and \emph{(2)}
    assign the third-to-last job to the \emph{second} least loaded processor, then assign the other two jobs as in \emph{(1)}.
    %\begin{enumerate}
    %    \item Assign each of the three final jobs to the least loaded processor respectively.
    %    \item Assign the third-to-last job to the second least loaded processor, then assign the other two jobs as in \emph{(1)}.
    %\end{enumerate}
\end{pruning rule}

\begin{pruning rule}\label{pruning_rule:fewer_jobs_than_procs}
    If $i < m$ jobs remain unassigned, then only the $i$ least loaded processors need to be considered for the next decision.
\end{pruning rule}

% Note that our wording of R\ref{pruning_rule:interchangeable_jobs} differs slightly from the previously presented version~\cite{dellamico_martello_1995} but remains equivalent.
These are the only rules from the literature that we use in this paper; for more rules we refer to the literature on the subject~\cite{dellamico_martello_1995, lifting_2, dissertation_lawrinenko}.

\subsection{New Pruning Rules.} % inter, FUR, RET (mostly SPAA) + Nikolai's rules

%We now present new pruning rules. % by considering the \pcmax{} \emph{decision} problem.
%\todo{DS: State somewhere that the decision problem is the same as bin packing, and provide some refs!}
In the following, we consider a \textit{decision} problem instance $(W, m, U)$ and a feasible partial assignment $A$ at decision level $\decisionlvl$.
We use $i := \decisionlvl + 1$ as a shorthand for referring to the job $j_i$ which is assigned next.
% \paragraph{Small Jobs}
%Depending on the distribution of the jobs, there might be a large number of small jobs with equal duration.
%The first new rule we present 
We first present a new rule that
allows us to efficiently handle the base case where all remaining jobs share the same duration.
 %handle large numbers of equally small jobs.
\begin{pruning rule}\label{pruning_rule:equal_remaining_jobs}
    %\todo{Suggestion, just all remaining jobs, so if we already for some reason assign $j_n$, but $w_l+1 = \ldots = w_{n-1}$ then this still applies.}
    %\todo{NM: not sure how to do this without first defining a set of remaining jobs, which seems a bit cumbersome to me.}
    If all unassigned jobs have equal duration (i.e., $w_{\decisionlvlplus} = \cdots = w_n$),
    then there is a valid completion of $A$ if and only if
    %\[
    %\sum_{x=1}^m \big( (U - C_x^{\decisionlvl}) \modnospace w_n \big) \le mU - \sum_{i=1}^n w_i.
    %\]
    \begin{equation}
        \label{eq:equal_remaining_jobs}
        \sum_{x=1}^m \left\lfloor \frac{U - C_x^A}{w_n} \right\rfloor \ge n - \decisionlvl
    \end{equation}
    %In this case, \wlogen{} only a single decision $a_{\decisionlvl+1} = x$ with $C_x^{\decisionlvl} + w_n \le U$ must be considered.
    If there is a valid completion for $A$, then such a completion follows from iteratively assigning each remaining job to a processor $p_x$ which satisfies $C_x^A + w_n \le U$.
\end{pruning rule}

\begin{proof}
    % Note that $n - \decisionlvl$ is the number of unassigned jobs.
    In the following, we use $s_x := \left\lfloor \frac{U - C_x^A}{w_n} \right\rfloor$ for the number of ``slots'' on processor $p_x$.
    If Equation~\ref{eq:equal_remaining_jobs} holds, we can choose a completion $S$ of $A$ in such a way that for each processor $p_x$, at most $s_x$ additional jobs are assigned to $p_x$.
    This implies $C_x^S \le C_x^A + s_x w_n \le U$.
    For the reverse direction, consider any valid completion $S$ of $A$ and let $n_x$ be the number of additional jobs assigned to processor $p_x$.
    $C_x^S \le U$ implies $n_x \le s_x$ and therefore $n - \decisionlvl = \sum_{x=1}^m n_x \le \sum_{x=1}^m s_x$.
\end{proof}

% old variant based on modulo:
% Consider a processor $p_x$ with load $C_x^{\decisionlvl}$.
% If we assign as many of the remaining jobs as possible to $p_x$,
% the new load is $C_x^{\decisionlvl + k} = C_x^{\decisionlvl} + k \cdot w_n$
% for the largest $k$ satisfying $C_x^{\decisionlvl + k} \le U$.
% The remaining space on $p_x$ is then $U - C_x^{\decisionlvl + k} = (U - C_x^{\decisionlvl}) \modnospace w_n$.
% This inevitable ``gap'' to the optimal load implies a lower bound for the \emph{total} remaining space of any feasible completion.
% Specifically, a completion with makespan $U$ is feasible if and only if the sum of all such gaps does not exceed the free space that must remain % mathematically -> ?
% after successfully assigning all jobs.
% new variant
% Our reasoning is that we can simply
%Intuitively, we compute the number of ``slots'' of size $w_n$ left on each processor and compare it to the number of remaining jobs.
Intuitively speaking,
the exact mapping of jobs to slots is inconsequential, %Assignments to these slots are symmetric in the sense that the jobs are freely swappable,
thus only one assignment needs to be considered.
%In Appendix~\ref{appendix:two_durations}, we extend this idea to \textit{two} remaining job durations.
%This idea can be extended to 
%an efficient dynamic programming solution for the case of \textit{two} remaining job durations,
%which we present in Appendix~\ref{appendix:two_durations}.
We can also restate this rule for the \pcmax{} \textit{optimization} problem as follows:
%For the \pcmax{} optimization problem., we can also restate this rule as follows. 
%Given a partial assignment $A$ where all remaining jobs have the same duration,
If all remaining jobs have the same duration,
an optimal completion %of $A$
can be obtained via the LPT algorithm.

%Therefore $U - C_x^{\decisionlvl + k} = (U - C_x^{\decisionlvl}) \modnospace w_n$ and we can use this to calculate a lower bound for the total remaining space of any feasible completion.
%Consequently, a completion with makespan $U$ is possible if and only if this bound does not exceed the actual available space.
% This rule allows to determine that a partial assignment is infeasible without any further branching.

Next, we provide stronger rules for breaking symmetries between processors.
Assume that we already assigned $u < U$ work to a certain processor and now decide whether to assign $j_\decisionlvlplus$ to this processor.
% (with $i = \decisionlvl + 1$).
% \todo{NM: this $i = \decisionlvl + 1$ situation is a bit weird IMO}
%Consider an instance $(W, m, U)$ of the \pcmax{} decision problem
%where we already assigned $u < U$ work to a certain processor and now decide whether to assign $j_i$ to this processor as well. 
Given
the set $J_\decisionlvlplus := \{j_\decisionlvlplus, \ldots, j_n\}$ of \emph{smaller jobs}, we define the function
\[
\phi(\decisionlvlplus, u) := \{ J' \subseteq J_\decisionlvlplus \mid u + \sum_{j_k \in J'} w_k \leq U\}.
\]

Intuitively, the function $\phi$ lists all possible combinations of jobs we can assign to a processor to still have an assigned workload (on this processor) $\leq U$.
With that we obtain the following pruning rules.

\begin{pruning rule}\label{pruning_rule:inter}
    %If there are processors $p_x, p_y$ with assigned workload %$C_x^\decisionlvl$ and $C_y^\decisionlvl$ such that
    %Given an instance \pcmaxdecisioninstance{} of the decision problem of \pcmax{} and a partial assignment $A$ with $|A|=\decisionlvl$, if there are processors $p_x, p_y$ with assigned workload $C_x^\decisionlvl$ and $C_y^\decisionlvl$ such that
    %\[
    %$
    %\phi(\decisionlvl, C_x^\decisionlvl) = \phi(\decisionlvl, C_y^\decisionlvl),
    %$
    %\]
    %then \wlogen{} we can prune the decision $a_{\decisionlvlplus}=y$.
    
     If there are multiple processors $P := \{p_x, p_y, \ldots\}$ such that $\phi(\decisionlvlplus, C_x^A) = \phi(\decisionlvlplus, C_y^A) = \ldots$,
     then only \emph{one} processor in $P$ has to be considered when assigning $j_\decisionlvlplus$.
\end{pruning rule}

\begin{proof}
    %We show that for any two $p_x, p_y \in P$, a feasible 
    Consider a feasible completion $S$ of $A$ where $j_i$ is assigned to $p_x \in P$.
    Let $J_x$ be the jobs that are assigned to $p_x$ in $S$ and are not already assigned in $A$, and analogously let $J_y$ be the jobs assigned to some other processor $p_y \in P$.
    Since $J_x \in \phi(\decisionlvlplus, C_x^A) = \phi(\decisionlvlplus, C_y^A) \ni J_y$ by definition of $\phi$, swapping $J_x$ and $J_y$ also results in a feasible solution.
\end{proof}

This is a generalization of R\ref{pruning_rule:equiv_old} which uses the $\phi$ function for even stronger symmetry breaking.
%The reasoning behind this rule is that if two processors have equal $\phi$ sets for a given partial assignment, then any completion of one processor can also be used as a completion of the other.
%Therefore, this rule breaks symmetry and preserves feasibility of a given \pcmax{} decision instance.

\begin{pruning rule}[The Fill-Up Rule, \emph{FUR}]\label{pruning_rule:fur}
    \todo{NM: should we remove the ``largest unassigned'' qualification, since it is redundant to $\phi(\decisionlvlplus, C_x^A) = \phi(\decisionlvlplus, U - w_\decisionlvlplus)$? MA: You are right here, it used to not be redundant because in practice we apply the FUR 'out of order' but because it can be more effective. Here because we change $j_i$ to refer to '$j_{l+1}$' it is superfluous.}
    If $j_\decisionlvlplus$ is the largest unassigned job that can still be assigned to processor $p_x$ (i.e., $C_x^A + w_\decisionlvlplus \leq U$) and\\
    % \[
    % w_\decisionlvlplus = \max \left\{
    %     \sum_{j_k \in J'}{w_k} \mid J' \in \phi(\decisionlvl, C_x^\decisionlvl)
    % \right\}
    % \]
    $\phi(\decisionlvlplus, C_x^A) = \phi(\decisionlvlplus, U - w_\decisionlvlplus)$,
    then
    we only need to consider the decision $a_\decisionlvlplus = x$.
    % we can assert $a_\decisionlvlplus = x$.
    % Note: "we can assert" is not a meaningful statement here - NM
\end{pruning rule}
\begin{proof}
    Consider a completion $S$ of $A$ where $j_\decisionlvlplus$ is assigned to $p_y$, with $y \neq x$.
    Let $J'$ be the jobs that are assigned to $p_x$ in $S$ and are not already assigned in $A$, and let $w'$ be their summed duration.
    $J'$ being assigned to $p_x$ requires that
    $J' \in \phi(\decisionlvlplus, C_x^A) = \phi(\decisionlvlplus, U - w_\decisionlvlplus)$
    by definition of $\phi$, which in turn implies $w' \le w_\decisionlvlplus$.
    Therefore, swapping $j_\decisionlvlplus$ with $J'$ results in a solution $S'$ where $C_x^{S'} \leq U$, $C_y^{S'} \leq C_y^{S} \leq U$ and $j_\decisionlvlplus$ is assigned to $p_x$.
    Thus, assigning $j_\decisionlvlplus$ to $p_x$ maintains feasibility.
    %Let us consider all of the jobs $j'_1, \ldots $ with weights $w'_1, \ldots $  that are assigned to $p_x$ and are not already assigned in $A$.
    %Since $j_\decisionlvlplus$ is the largest job that can still be assigned to $p_x$, we know that $\{j'_1, \ldots \} \in \phi(\decisionlvlplus, C_x^A) = \phi(\decisionlvlplus, U - w_\decisionlvlplus)$, implying that $\sum_k{w'_k} \leq w_i$.
    %With that, we know that by swapping $j_\decisionlvlplus$ with $j'_1, \ldots $ we obtain a solution $S'$ where $C_x^{S} \leq C_x^{S'} \leq U$ and $C_y^{S'} \leq C_y^{S} \leq U$.
    %Thus, we can directly assign $j_\decisionlvlplus$ to $p_x$ and maintain feasibility
\end{proof}

In addition to this, the rule enables pruning when the upper bound is updated.
This is described in the following theorem.

\begin{theorem}\label{theorem:fur_upper_bound}
    Consider the case that the FUR % R\ref{pruning_rule:fur}
    applies to job $j_\decisionlvlplus$ and processor $p_x$
    and let $A' := A \cup \{ a_\decisionlvlplus = x \}$ be the partial assignment created by it.
    If an optimal completion of $A'$ has makespan $U'$ such that $U' > C_x^A + w_\decisionlvlplus$,
    then it is also an optimal completion of $A$.
\end{theorem}

\begin{proof}
    Assume there is a feasible completion of $A$ for the decision instance $(W, m, U' - 1)$.
    %with makespan $U' - 1$.
    Then $A$ is a valid partial assignment and since $C_x^{A} + w_\decisionlvlplus \le U' - 1$, the FUR is applicable to $p_x$.
    The correctness of the FUR implies that if a feasible completion exists for $A$,
    a feasible completion (with makespan at most $U' - 1$) also exists for $A'$, which is a contradiction.
\end{proof}

%This theorem is derived by applying the FUR to the decision instance with makespan $U' - 1$.
%Assume $A$ has as completion with makespan $U' - 1$.
%Then $A$ is a valid partial assignment and since $C_x^{\decisionlvl} + w_i \le U' - 1$, the FUR is applicable to $p_x$.
%The correctness of the FUR then implies that if a valid completion exists for $A$,
%a valid completion also exists for $A'$, which is a contradiction.

% We can use this theorem to speed up our \bnb{} algorithm by foregoing further branching in cases where the FUR leads to a solution which this theorem says is an optimal completion of the current partial assignment.

The effectiveness of the rules that use the $\phi$ function highly depends on the sizes of the remaining jobs.
% For example, if an instance includes a job of size one,
% the $\phi$ sets of two processor are equal only if both processors have equal load.
In the following, we show that the smallest jobs can be removed from the instance in some cases,
thereby enabling stronger symmetry breaking via the $\phi$ function.
%Apart from reducing branches, this also allows for stronger symmetry breaking in combination with the other new rules presented afterwards.

\begin{theorem}\label{theorem:smallest_job}
    If the following inequality holds,
    \[
        \sum_{i=1}^{n-1} w_i < m \left( U - w_n + 1 \right),
    \]
    the instance $(W, m, U)$ is equivalent to $(W \setminus \{ w_n \}, m, U)$, where the smallest job is removed.
\end{theorem}

\begin{proof}
     For any valid partial assignment $A$ with depth $n - 1$
     (i.e., solution of the new instance),
     %the inequality implies there is a processor $x$ with load $C_x^{n-1} \le U - w_n$
     dividing the inequality by $m$ shows that there is at least one processor $p_x$ with load $C_x^{n-1} < U - w_n + 1$ and thus $C_x^{n-1} \le U - w_n$.
     Therefore, assigning $w_n$ to $p_x$ yields a solution of the original instance.
\end{proof}

Note that this theorem may apply again for the modified instance, allowing to remove multiple jobs.
%Therefore, multiple jobs might be removed by applying it repeatedly.
%Note that applying this theorem repeatedly might remove multiple jobs from an instance.

\subsection{Efficient Computation.}

Since $\phi$ encompasses a combinatorial co-domain, computing it explicitly would be prohibitively expensive.
Therefore, we introduce an auxiliary data structure named \emph{range equivalency table (RET)}, which allows us to \emph{implicitly} determine whether $\phi(i, C_x^A) = \phi(i, C_y^A)$.

Let us first highlight two interesting properties of $\phi$.
First, $\phi(i, u) \supseteq \phi(i, u+1)$ for all $u < U$.
This is %due to the fact that, in layman's terms, if you obtain more space to insert jobs, you can insert at least as many jobs as you previously could.
because the sets of admissible jobs will never decrease when increasing the available processing time.
%extending a processor's admissible processing time can only increase 
Secondly, for $i < n$, the valid ways to assign jobs that may or may not include $j_i$ is equal to the union of valid ways of assigning jobs that do not include $j_i$, and the valid ways of assigning jobs that do:
\begin{align}
\begin{split}
\phi(i, u) =\ 
    & \phi(i+1, u) \\
    & \cup \{j_{i} \cup X \mid X \in \phi(i+1, u + w_{i})\}
    \label{eq:inductive-ret}
\end{split}
\end{align}
%This property is also relatively clear, since 

%For an instance of the \pcmax{} decision problem \pcmaxdecisioninstance{}, 
We define the RET as an $n \times (U + 1)$ table with entries in $\N$---%
%To better match the purpose of the RET, we use a rather peculiar method of indexing.
%The first dimension of the RET is indexed from $1$ to $n$---
one row per job
%The second dimension of the RET is indexed from $0$ to $U$%, since each entry represents the size of the assigned workload.
%---
and one column for each possible assigned workload.
Note that individual values in the RET have no inherent meaning other than serving as identifiers (IDs) for equivalent $\phi$ sets within the same row.
%\todo{Make it clear that the numbers stored in the RET are in a sense, arbitrary, we do not care what the number is, as long as it is different between ERs
%- NM: attempt at beginning of next paragraph
%-- DS: reformulated}
For a job $j_i$, an \emph{equivalence range} %in the RET 
is a range of workloads $u, \ldots, u'$ such that $RET[i][u] = \ldots = RET[i][u']$.
For such a range we assert that $\phi(i, u) = \ldots = \phi(i, u')$.
This means that, regardless of whether a processor's load is $u$ or $u'$ (or in between), the possible combinations of jobs from $j_i$ onward which it can still be assigned %to the processor 
are exactly the same.

% Nikolai:
%The singular property of the RET is that entries are equal if and only if the according $\phi$ values are equal---note that the stored values have no inherent meaning other than that.
%
We construct the RET going from the smallest ($j_n$) to the largest job ($j_1$).
%For $j_n$, $\phi(j_n,u) = \{ \emptyset, \{j_n\} \}$ if $u + w_n \leq U$, and $\phi(j_n,u) = \{ \emptyset \}$ otherwise.
For $j_n$, if $u + w_n \leq U$ then assigning $j_n$ still constitutes a valid workload; $\phi(n,u) = \{ \emptyset, \{j_n\} \}$. Otherwise, assigning no further job is the only option; $\phi(n,u) = \{ \emptyset \}$.
We thus initialize two equivalence ranges: 
%$RET[n][U -w_n+1]=\ldots=RET[n][U]=1$ and $RET[n][0]=\ldots=RET[n][U-w_n]=2$.
$RET[n][u]=1$ for $U-w_n < u \leq U$ and $RET[n][u]=2$ for $0 \leq u \leq U-w_n$.
For job $j_i$ ($i<n$), we 
%assume that the RET has been built for $j_{i+1}$.
%We again 
denote the two relevant prior entries for applying property (\ref{eq:inductive-ret}) as
$\funcleft(i, u) := RET[i+1][u]$ and $\funcright(i,u) := RET[i + 1][u+w_i]$ (with $\funcright(i, u) := 0$ if $u+w_i > U$).
Note that going from some $u$ to $u-1$, the equivalence range for $j_i$ should remain unchanged if and only if both $\funcleft$ and $\funcright$ remain unchanged.
As such, we start by setting $RET[i][U]:= 1$ and then proceed sequentially for $u=U-1,U-2,\ldots,0$:
\[
RET[i][u] := \begin{cases}
    RET[i][u+1], \\
    \hspace{5mm} \textrm{if } \funcleft(i, u) = \funcleft(i, u+1) \\
    \hspace{5mm} \phantom{\textrm{if }} \wedge \funcright(i, u) = \funcright(i, u+1); \\
    RET[i][u+1] + 1, \\
    \hspace{5mm} \textrm{otherwise.}
\end{cases}
\]
%In words, we initialize a new equivalency range for $j_i$ if either of $\funcleft$ or $\funcright$ changes from $u+1$ to $u$, and we extend the prior range otherwise.

\begin{restatable}{theorem}{rettheorem}\label{thm:ret}
    Given an instance \pcmaxdecisioninstance{} of the \pcmax{} decision problem, for any $i \in \{1,\ldots, n\}$ and $u , u'< U$:
    %\[
    $RET[i][u] = RET[i][u']$
        if and only if
    $\phi(i, u) = \phi(i, u')$.
    %\]
\end{restatable}
%The reasoning for this theorem is that the value of an $RET$ entry changes iff $\phi(i, u) \neq \phi(i, u+1)$ or $\phi(i, u + w_i) \neq \phi(i, u + w_i+1)$.
\begin{proof}
    By induction.
    For $i = n$ this property is clear.
    For $i < n$, we assume that Theorem~\ref{thm:ret} holds for $i+1$.
    W.l.o.g., assume $u < u'$.
    If $RET[i][u] = RET[i][u']$, then we know $RET[i+1][u] = RET[i+1][u']$ and $RET[i+1][u + w_i] = RET[i+1][u' + w_i]$.\footnote{This is due to the fact that if $RET[i][u] = RET[i][u']$ then by construction $\funcleft(i,u) = \funcleft(i,u+1) = \ldots = \funcleft(i,u')$ and $\funcright(i,u) = \funcright(i, u+1) = \ldots = \funcright(i,u')$.}
    Via Theorem~\ref{thm:ret} for $i+1$, this implies that 
    \begin{align*}
        \phi(i, u) &=  \phi(i+1, u) \cup \{j_i \cup x \mid x \in \phi(i+1, u + w_i)\} \\ 
        &=  \phi(i+1, u') \cup \{j_i \cup x \mid x \in \phi(i+1, u' + w_i)\} \\
        &= \phi(i, u').
    \end{align*}
    If $RET[i][u] \ne RET[i][u']$, then $\funcleft(i, u) \ne \funcleft(i, u')$ or $\funcright(i,u) \ne \funcright(i,u')$.
    If $\funcleft(i, u) \ne \funcleft(i, u')$, then, via Theorem~\ref{thm:ret} for $i+1$, $\phi(i+1, u) \supsetneq \phi(i+1, u')$.
    Therefore, $\phi(i,u) \supsetneq \phi(i,u')$ since all other elements in either set must contain $j_i$.
    The case where $\funcright(i,u) \ne \funcright(i,u')$ proceeds analogously.
    %Induction thus proves both directions in the above theorem.
    \todo{review comment on the definition/proof:\\
    - It is not clear to me while RET[i][u] = RET[i][u+1] + 1 in the 'otherwise' case. Why can the number stored in RET[i][u] not increase by more than one compared to RET[i][u+1]?
    -- DS: hopefully addressed by better explanation of the RET cells' meaning\\
- line 200 relies on the monotonicity of RET[i][u] without explicitly mentioning it.
    -- DS: Does it? I think it's now fine as is.
    }
\end{proof}

%The space requirements of the RET are in $\mathcal{O}(U\cdot n)$.
%\todo{NM: sentence seems not really necessary}
%Note that we can compress this further to $\mathcal{O}(U)$ by only storing at index $u$ the index of the smallest job $i$ where $RET[i][u] \neq RET[i][u+1]$.
%While we do not discuss this technique in detail, we found it to significantly cut our scheduler's memory usage for large instances.

\subsection{SAT-inspired Learning of Dead-ends.}

\begin{figure}[b!]
    \centering
    \includegraphics[width=0.7\columnwidth]{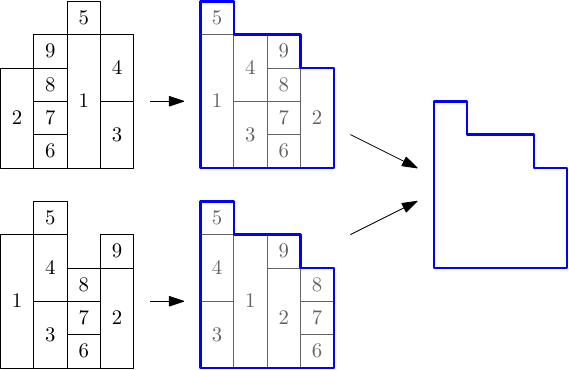}
    \caption{The two assignments on the left are equivalent when re-ordering the machines and then focusing on the \textit{silhouette}.
        Numbers correspond to job indices.
        %\todo{NM: perhaps use a more vertical version of the image, i.e., using two rows and converging to the right? DS: How's that?}
    }
    \label{fig:silhouette}
\end{figure}

For our next contribution, let us examine a crucial feature of modern SAT solvers, namely the \emph{clause learning} performed by the CDCL algorithm~\cite{marques2021conflict}.
%In this section, we outline the structure of a novel \bnb{}-like algorithm for the \pcmax{} problem.
%The CDSM algorithm is an extension of the DMFUR algorithm, where we now take some inspiration from the CDCL algorithm for SAT solving.
When a solver encounters a logical conflict during its search through the space of partial assignments, it detects 
%In the CDCL algorithm, after a conflict is reached the current partial assignment is analysed.
%We heuristically detect 
a subset of the current partial assignment that causes this conflict and derives a new clause from it that can be added to the formula.
%this partial assignment we have chosen to encode 
The SAT solver's logic now ensures that the sub-space triggering the found conflict
will never be revisited (except if the learned clause is discarded).

We %attempt to
transfer this technique to our \bnb{} algorithm for \pcmax{}. %decision problem.
Given a partial assignment that leads to a conflict, i.e., an infeasible assignment,
we ensure that we do not visit the partial assignment again.
%we extract the relevant part of the partial assignment and ensure that we do not visit it again.
%In order 
To represent 
%In the context of \bnb{}, we refer to 
the relevant information of the current partial assignment in a compact manner, we introduce the notion of \emph{silhouettes}:
Given a partial assignment $A$ that assigns jobs $j_1, \ldots, j_\decisionlvl$, we define the silhouette of $A$ as the multiset $\{C_1^A, \ldots, C_m^A\}$.
%In the context of the \pcmax{} problem,
It is clear that given two partial assignments with the same silhouette, one can be completed to obtain an optimal solution if and only if the other one can as well.
%\todo{review comment: I would appreciate it if you mentioned explicitly that the same silhouette implies that the partial schedules are for the \emph{same} set of jobs}.
Note here that two partial assignments having equal silhouettes is only possible if they are defined on the same subset of jobs.
Fig.~\ref{fig:silhouette} provides an illustration.
In the context of the \pcmax{} decision problem, we can further generalize this notion using the $\phi$ function.
Instead of storing the makespan of each processor at the current state, we store its current equivalence range.
We define the \textit{Generalised Silhouette (gist)} as the pair $(G, i)$ where $G$ is the multiset $\{ RET[i][C_1^A], \ldots, RET[i][C_m^A] \}$
and $i$ is the index of the first unassigned job $j_i$.
\todo{DS: This last sentence has been missing for some reason, I added it back with (I think) the correct $i$ index. I assume it was deleted in err, since we need the definition of a gist?}

% ==================== With the theorem
% \begin{theorem}\label{theorem:cdsm}
%     Given two partial assignments $A, B$, on the first $i$ jobs of a \pcmax{} decision instance \pcmaxdecisioninstance{}, if $A$ and $B$ have the same gist, then $A$ can be completed to obtain a satisfying solution if and only if $B$ can be as well.
% \end{theorem}

% \begin{proof}
%     Let us assume w.l.o.g. that $ \phi(i, C_1^A) = \phi(i, C_1^B), \ldots, \phi(i, C_m^A) = \phi(i, C_m^B)$.
%     Given any solution $S$ that is a continuation of $A$ we set $S' := B \cup \{a_\alpha = x \in S \mid \alpha > i\}$ to obtain another solution that is a continuation of $B$. 
% \end{proof}
% \todo{NM: If we need to save more space, I think we could remove the theorem and just write two or three sentences. The statement is really intuitive in my opinion.}

% We refer to our \bnb{} algorithm which exploits Theorem~\ref{theorem:cdsm} as \emph{Conflict Driven Silhouette Memory (CDSM)}.
%=======================================================
% 
% 
% 
%========================== With the paragraph
If two states have the same gist, then they are equivalent in the context of the \pcmax{} decision problem:
If one of the two can be completed to obtain a feasible solution, then the other one can as well. 
We refer to our \bnb{} algorithm which exploits this property as \emph{Conflict Driven Silhouette Memory (CDSM)}.
%============================================
% 
% 
% 
%Thus, in essence, during a \bnb{} algorithm for the \pcmax{} decision problem we only need to branch on each gist once. 
%If we revisit a gist we have previously visited, we know that it cannot be used to obtain an improvement on our current best solution.
%Thus, in the \emph{Conflict Driven Silhouette Memory (CDSM)} algorithm we run the DMFUR algorithm while 
CDSM memorizes %as many 
visited gists and prunes all decisions that would lead to a memorized gist,
%All decisions that would lead to such a memorized gist can then be pruned, 
since an equivalent partial assignment has already been visited.
In other words, CDSM needs to branch on each stored gist only once.

\section{Bounding Techniques}\label{sec:bounds}
% improving the lower and upper bounds from literature

%Upper and lower bounding techniques are highly effective at solving \pcmax{} exactly in many cases.

Upper and lower bounding techniques can solve many \pcmax{} instances exactly and accelerate exact solving in many other cases~\cite{lifting_2}.
%For instance, Haouari et al.\ remark that their results on bounding techniques ``\emph{raise the legitimate question of whether \pcmax{} could be considered as rather easy from the practical computational point of view}''~\cite{lifting_2}.
We now outline our improved bounding methods.

\subsection{Lower Bounds.}\label{subsec:lower_bounds_new}

Good lower bounds are crucial to solve ``simple'' instances without invoking an exact solver.
%Additionally, they sometimes allow to skip proving the infeasibility of a decision instance---a co-NP-complete task which is often prohibitively costly.
Moreover, they can %at times 
take the burden from a complete solver to ``prove'' that a certain decision instance is infeasible---a co-\textit{NP}-complete task which can be prohibitively costly in the general case.
As such, lower bounds can drastically accelerate \bnb{} approaches.

%Good lower bounds are especially crucial for SAT-based \pcmax{} scheduling %, a tight lower bound is crucial, 
%since they can allow to skip proving a formula's \emph{unsatisfiability}---a co-NP-complete task which is often prohibitively costly.
%Other \pcmax{} approaches benefit from these techniques as well.

%We thus spent some effort on improving lower bounding techniques, which, 
%again, benefit any \pcmax{} approach.
%We improve the performance of an existing \emph{lifting procedure}. % presented in Section~\ref{subsec:lower_bounds_lit}.
%It is thus usually worth it to invest more time into calculating a better lower bound, since it is, oftentimes, cheaper to calculate a better bound than it is to prove the unsatisfiability of a given instance.

Haouari et al.~\cite{lifting_1} show that for a given \pcmax{} instance, there are $O(n)$ easily obtainable sub-instances for which a lower bound can provide a close-to-optimal lower bound on the original instance.
%We refer to the original publication for the exact manner in which these instances are generated.
%This procedure is outlined in some more detail in Appendix~\ref{appendix:lifting}.
We propose to take this method a step further by solving these sub-instances \emph{optimally}, which yields tighter bounds on our original problem if done correctly.

\paragraph{Lift++.}
Given an instance \pcmaxinstance{} of the \pcmax{} problem and initial upper and lower bounds $U,L$,
we use the procedure presented by Haouari et al.\ to generate sub-instances, and then we bound them from above and below.
%and remember the best lower bound $L$ on our original instance.
We then attempt to optimally solve any sub-instances that might improve upon
the best lower bound $L$ on our original instance.
%After all sub-instances have been generated and bound, we can attempt to solve any instance that might improve upon $L$.
For any sub-instance with %upper and lower
bounds $U',L'$ such that $U'>L$ we try to solve it exactly within a given time frame.
We choose our time frame in such a way that solving each instance can be attempted within the time allocated to this bounding technique.
Preliminary tests show that in practical scenarios, the number of such instances tends to be very small (0--5).
These instances also tend to be very small, meaning they can be solved exactly almost all the time.
%As we see in Section~\ref{subsec:exp_bounding} this technique warrants its extra computational complexity.
%Our experiments suggest that this technique warrants the invested effort.

\subsection{Upper Bounds.}\label{subsec:upper_bounds_new}

%In the context of SAT solving, upper bounds decrease the number of SAT calls necessary to solve a given instance.
%For \bnb{} algorithms, the initial upper bound decides how aggressively search space can be pruned.
%Any heuristic algorithm for the \pcmax{} optimization problem yields an upper bound for the considered instance.
Upper bounds reduce the initial search space and thereby reduce the work that an exact solver needs to perform.
It is beneficial to employ a series of upper bounding procedures of varying cost and accuracy:
A cheap upper bound sometimes solves a \pcmax{} instance exactly, which allows to use more expensive %upper bounding
techniques more sparingly.
%Again, our upper bounds are usually cheaper than SAT calls, so it is often worth the extra computational cost.
%In this section, 
%We now present two novel upper bounding techniques for the \pcmax{} problem.

\begin{figure*}[t]
    \centering
    \includegraphics[width=0.85\textwidth]{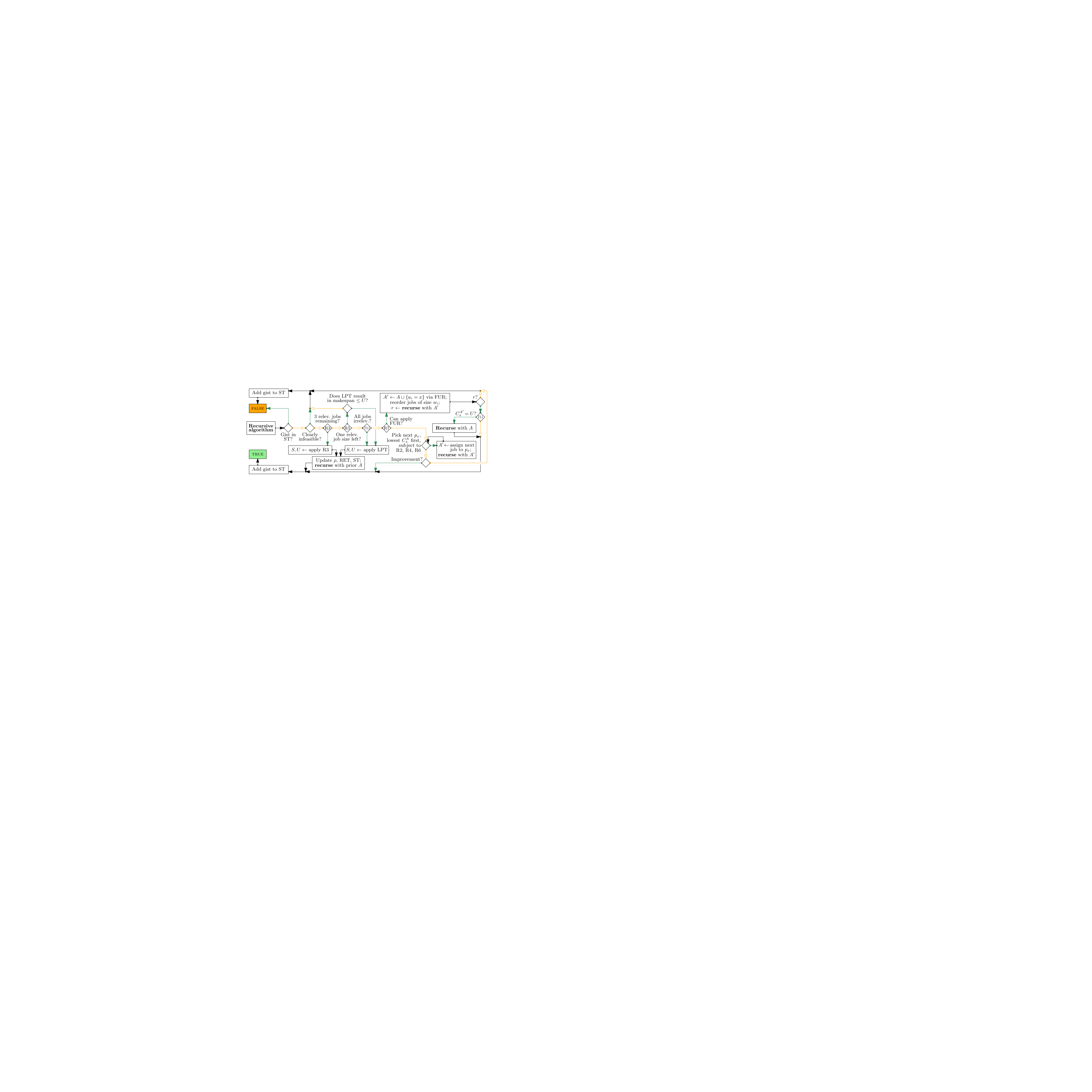}
    \caption{
        \bnb{} algorithm.
        Diamonds represent ``if-then-else'' nodes; green filled arrows represent \textsc{true}-branches, orange non-filled arrows represent \textsc{false}-branches.
        Where applicable, diamonds are labelled with the rule or the theorem enabled by the respective condition.
        Global fields common to all recursive calls are $U$ and $L$, irrelevance index \algirrelevanceindex{}, the \algret{}, state table \algstatetable{}, and best solution \algbestsolution{}.
    }
    \label{fig:flowchart}
\end{figure*}

\paragraph{LPT++.}
We propose to use a modified version of the Fill-Up Rule, which does not require the RET, for improving upper bound heuristics.
Given the largest unassigned job $j_i$ and a processor $p_x$ with remaining space exactly $w_i$,
we can assign $j_i$ to $p_x$ and maintain feasibility.
The \emph{LPT++ heuristic} for \pcmax{} decision instances then works by combining the LPT algorithm with this rule.
To obtain a new upper bound, we execute LPT++ for makespans within a given initial lower and upper bound.
The lowest value where LPT++ succeeds is an upper bound for the original \pcmax{} instance.

\iffalse
Our pruning rules presented earlier can also be used to improve the results of LPT.
Thus, we define the \emph{LPT++ heuristic} for the \pcmax{} decision problem as follows.
We combine the LPT heuristic with a special case of the Fill-Up Rule which does not require constructing the \emph{RET}.
%\psfrage{It does not get clear whether rule is only applied to the job currently considered by the LPT order or to any job that happens to fit. Note that the latter would resemble the well-known best-fit algorithm.}
%Since the pre-computation required to apply the Fill-Up Rule can be quite time-consuming, we use the simplified Fill-Up Rule (R\ref{pruning_rule:fur_simp}).
In essence, this rule states that given a partial assignment where processor $p_x$ has remaining space exactly $w_i$, where $j_i$ is the largest unassigned job, we can simply assign $j_i$ to $p_x$ and maintain feasibility. %\psfrage{In the context of a scheduling heuristics, it would be better to call this feasibility rather than satisfiability?}
%\todo{[NM: Yeah we also have ``feasible'' and ``valid'' for partial assignments, we should unify the terminology]}.
%
In order to obtain an upper bound for a given \pcmax{} optimization instance and initial lower and upper bounds $(L, U)$, we execute LPT++ on the respective \pcmax{} decision problems for makespans within $[L, U)$.
The lowest value for which LPT++ succeeds is a valid upper bound for the original \pcmax{} instance.
%For the sake of convenience, we refer to this new \pcmax{} heuristic as LPT++.
\fi

\paragraph{MS++.}
Our \emph{MS++ heuristic} is loosely inspired by the \emph{Multi-Subset} heuristic by Dell'Amico and Martello~\cite{dellamico_martello_1995}.
%In the MS++ heuristic, a \pcmax{} decision problem is given with a makespan $U$.
For each processor $p_x$ of a \pcmax{} decision problem, we solve a \emph{Subset-Sum-Problem} (SSP) with the remaining jobs---%
%To each processor, we 
sequentially assigning the subset of jobs that maximises $C_x$ while maintaining $C_x \leq U$.
This can be done in pseudo-polynomial time using dynamic programming~\cite{pisinger1999linear}.
Similar to the LPT++ heuristic,
we then obtain a new bound by executing the MS++ heuristic for makespans within some initial bounds.
%In order to bound a given \pcmax{} instance we then, similar to the LPT++ heuristic, execute the MS++ heuristic on the respective \pcmax{} decision problems with makespans in $[L, U)$ and return the lowest value for which the MS++ heuristic returns true.
%Similar to before, we also refer to this new \pcmax{} heuristic as MS++ for the sake of convenience.

\paragraph{S$^4$.} The \emph{Single Start Subset-Sum} heuristic is inspired by the \emph{Multi-Start Subset-Sum} (MSS) heuristic by Haouari and Jemmali~\cite{lifting_2}.
MSS starts with a random assignment and then repeatedly chooses two random processors to reschedule their jobs \textit{optimally} on those two processors via an exact SSP algorithm.
When a local minimum is reached, MSS restarts with a new random assignment.
%In the MSS heuristic, we generate a random complete assignment, then we repeatedly choose two processors at random, remove all of their jobs, and then reschedule them optimally on those two processors via an exact algorithm for SSP.
%This is done repeatedly until we reach a local minimum, where we restart with a new random assignment.
In the S$^4$ heuristic, we instead start with a high-quality assignment and improve it similarly to MSS.
To escape local minima, we use perturbations:
%However, instead of restarting at local minima,
%we perturbate the assignment by
We select a processor at random, remove $i \ll n$ of its jobs, and reschedule them randomly.
The perturbation factor is gradually increased until a timeout is reached.

\section{Implementation}\label{sec:implementation}

% Implementation may merit its own top-level section:
% efficient computation / maintenance of RET and the CDSM table
% also put the improvements to the ILP approach here, if it's something rather simple?

%We now outline our implementation, beginning with our \bnb{} framework.
%\subsection{\bnb{} Algorithm}

We implemented a single \bnb{} framework integrating all of our applicable techniques.
A flow chart for the full recursive algorithm (CDSM) is given in Fig.~\ref{fig:flowchart}.

%We choose to implement a simple recursive algorithm.
%Our recursive algorithm 
\algalgorithm{} takes a partial assignment $A$ (initially $A = \emptyset$) and accesses global variables for $W$, $m$, $L$, and $U$.
In addition, we use
\algbestsolution{} to represent the current best solution, whose makespan is $U+1$ (initialized via some upper bound),
%whose makespan we assert to be $U+1$ at the beginning of each call;
and \algirrelevanceindex{} for the job index where \emph{all later jobs are eliminated} by Theorem~\ref{theorem:smallest_job}.
The RET as well as the %CDSM's
\emph{State Table} (ST), which is implemented as a hash table, are also global objects.
%In addition, we introduce two further global variables:
%\algbestsolution{} represents the current best solution, initialized via some upper bound, whose makespan we assert to be $U+1$ at the beginning of each call;
%and \algirrelevanceindex{} is the job index where \emph{all later jobs are eliminated} by Theorem~\ref{theorem:smallest_job}.
Before calling \algalgorithm{} for the first time, we pre-compute $\rho$ and the RET (disregarding irrelevant jobs according to $\rho$) once.
\algalgorithm{} returns \algtrue{} iff it completed $A$ in makespan $\leq U$, in which case \algbestsolution{} is guaranteed to be an \emph{optimal} completion of $A$.
This invariant allows to solve the optimization problem with a single top-level call to CDSM
while taking full advantage of the pruning rules for the decision problem.
%\todo{review comment: explain better how a single call is possible (- NM: this seems to be counter-intuitive)}
Note that in some cases a new recursive call is still necessary after finding an improved solution, which is however handled internally in the CDSM procedure.

%Within \algalgorithm{}, we consider the decision instance defined on the given instance and $U$.
Whenever a solution is found, we update $U$, the \algret{}, and \algirrelevanceindex{}.
If \algirrelevanceindex{} changes, we lazily reconstruct the \algret{} from the bottom up, due to the newly considered jobs, until we detect that levels no longer change.
Note that we also need to clear the \algstatetable{} in this case.
If \algirrelevanceindex{} does not change,
%Instead, we can simply
it suffices to shift each entry in the RET to the left to match the decrease in $U$.
We implement this by incorporating an offset into each \algret{} query.
For each call to \algalgorithm{} with a previously unseen gist, we first check whether completing $A$ within $[L,U)$ is clearly infeasible:
We ignore processors where we can no longer insert any jobs and sum up the free space on the remaining processors.
If this space is less than the sum of the remaining job weights
%or if updating $L$ via the trivial lower bound~\cite{dellamico_martello_1995} \todo{NM: could we clarify what ``updating $L$ via the trivial lower bound'' concretely means? I assume something like applying the first two cases of the TVL to the remaining jobs and the processor with smallest load? This is not particularly obvious, however
%-- DS: was the appendix added after that remark? is it fine now, or otherwise add another sentence to the appendix?
%} (Appendix~\ref{appendix:trivial_bound}) 
% \todo{review comment: What is the trivial lower bound of [3] - NM: I would also appreciate an explanation}
%yields $L>U$,
%%%% Version 1 %%%%
or if the current partial assignment already has makespan $U$ (which is possible if $U$ was recently updated),
%%%%%%%%
the instance is infeasible. %and we return \algfalse{}.
%%%% Version 2 %%%%
%In addition, we update the lower bound to include the makespan of the current partial assignment for recursive calls.
%%%%%%%%
Otherwise, we %update $L$ via the trivial lower bound (returning \algfalse{} if $L>U$) and then 
apply our pruning rules where possible, excluding jobs that are irrelevant according to \algirrelevanceindex{}.\footnote{
    Note that most pruning rules do not interfere, with one exception.
    When R\ref{pruning_rule:fur} is applied to a job,
    this job is not eligible for R\ref{pruning_rule:interchangeable_jobs} anymore.
    We implement this by marking the according jobs.
} % beginning with R\ref{pruning_rule:last_levels} and then R\ref{pruning_rule:equal_remaining_jobs}.
%We also check if \textit{all} remaining jobs are irrelevant, 
%in which case we apply Theorem~\ref{theorem:smallest_job}, placing the remaining jobs via LPT. %to improve upon \algbestsolution{} %, update the \algret{}, \algirrelevanceindex{}, and $U$, 
%and recursing with the updated upper bound.
After applying R\ref{pruning_rule:last_levels}, R\ref{pruning_rule:equal_remaining_jobs}, or Theorem~\ref{theorem:smallest_job}, recursion may be required if $U$ is updated, since this possibly unlocks further improvements.
%Then, if R\ref{pruning_rule:fur} can be applied, we do so and recurse.
If we apply R\ref{pruning_rule:fur} and the initial recursion succeeds, we can either apply Theorem~\ref{theorem:fur_upper_bound} or else need to undo the FUR assignment to recurse with the new bound.
In the most general case, we branch over all processors not pruned by R\ref{pruning_rule:interchangeable_jobs}, R\ref{pruning_rule:fewer_jobs_than_procs}, or  R\ref{pruning_rule:inter}.
We first explore the branch that assigns the job to the least loaded processor,
since preliminary experiments showed this is more effective than selecting the branch based on the processor index.
\todo{C-R: cite thesis!}

\section{Evaluation}\label{sec:evaluation}

% i10pc126
We implemented our approach in Rust and used (A) an 80-core ARM Neoverse-N1 at 3\,GHz with 256\,GB of %DDR4 
RAM and (B) a 64-core AMD EPYC Rome 7702P at $\leq 3.35$\,GHz with 1\,TB of RAM.
Code and data are available at \url{https://github.com/matthewakram/P-Cmax-solver.git}.

\subsection{Benchmarks.}

%We consider 3500 problem instances defined by Mrad and Souyah~\cite{8248744} with $n/m \in [2, 3]$, $n \in [20, 220]$, and both uniform and normal distributions to generate job sizes.

We consider three kinds of benchmarks.
First and foremost, we use three different benchmark sets from literature, which consist of randomly generated instances:
\begin{itemize}
    \item \texttt{Lawrinenko}: 3500 benchmarks as described by Mrad and Souyah~\cite{8248744}, originating from unpublished work by Lawrinenko, inspired by the findings of Haouari and Jemmali~\cite{lifting_2}.
    \item \texttt{Berndt}: 3060 benchmarks generated as described by Berndt et al.~\cite{berndt2022load}.
    \item \texttt{Frangioni}: 780 benchmarks as described by Frangioni et al.~\cite{frangioni2004multi}.
\end{itemize}
Secondly, we assembled benchmarks based on real data from three application disciplines:
\begin{itemize}
    \item \textbf{Running times.}
    We constructed \pcmax{} instances from running times of certain applications.
    We use running times from SAT Competition 2022 (in s and s/10, each with one unit of upstart time added) and parallel SAT solving~\cite{schreiber2023scalable} (in s/10), bioinformatics tool \texttt{RAxML}~\cite{hoehler2021raxml,huebner2021exploring} (in s), MapReduce applications~\cite{hespe2023enabling} (in min), perfect hashing~\cite{lehmann2023sichash} (in $\mu$s), and queries for dynamic taxi-sharing~\cite{laupichler2024fast} (in ms)\nolinebreak
    \footnote{Even though these running times are \emph{a priori} unknown in reality,
    the underlying job size distributions can still be relevant, e.g., when comparing online schedulers to an optimal offline schedule of \emph{observed} running times~\cite{schreiber2021scalable} or when computing offline schedules based on \emph{predicted} running times~\cite{ngoko2019solving}.}.
    In cases where we sample job sizes from a large set of values and thus $n$ is variable, we used $n/m \in \{2, 2.5, \ldots, 4.5, 5\}$.
    \iffalse
    Even though most of the used running times are \emph{a priori} unknown in reality, testing \pcmax{} schedulers on these distributions is still sensible since they can be relevant, e.g., when comparing online schedulers to an optimal offline schedule of the \emph{observed} running times~\cite{schreiber2021scalable} or when computing offline schedules based on \emph{predicted} running times~\cite{ngoko2019solving}.
    \fi
    \item \textbf{Vertex degrees in a graph.}
    When processing graphs in distributed systems, a graph is commonly partitioned across the compute nodes,
    %while approximating the workload via edge counts.
    while using edge counts to approximate workload.
    %In cases where the locality of adjacent nodes is of no concern, each vertex' degree (i.e., edge count) can serve as a simple weight metric.
    We thus consider each vertex as a job whose size is equal to the vertex degree.
    \item \textbf{Clause lengths in a CNF.}
    An undisclosed researcher is considering distributed simplification of SAT formulas.
    Some simplification techniques can be performed for each problem clause of length $>1$ (e.g., clause strengthening), and these simplifications may be load balanced according to the total literals assigned to each processor. 
\end{itemize}
Thirdly, we consider \emph{planted instances}. %, i.e., problem instances where an optimal solution is known by construction but possibly hard to find.
Given parameters $(n, m, U)$, we begin with an optimal scheduling of $m$ jobs on $m$ processors, each with duration $U$.
We draw a pair $(x, t)$ of processor $x$ and time $0 \leq t < U$ uniformly at random and cut the job at this spot into two jobs.
This step is repeated until there are $n$ jobs.
We then increment the duration of $\lceil r \cdot n \rceil$ random jobs, for $r \in \{0, 0.01, 0.05, 0.1\}$, and output a \pcmax{} instance with the sorted job durations.
For each $m$, we considered job-to-processor ratios $n/m \in \{2, 2.5, 3, 4, 5, 7, 10\}$.

For our application and planted benchmarks, we construct \pcmax{} instances for numbers of processors
$m =$ $3$, $4$, $5$, $7$, $10$, $20$, $30$, $50$, $100$, $300$, $500$, $1000$, $2000$, $3000$, and we pre-filter instances where LPT combined with trivial lower bounds identifies an optimal solution.

% real-world instances needed !
%% 30k SAT instances. Job size is given by file size. Motivation: Schedule formula processing / feature extraction tasks in such a way that RAM requirements are balanced across all used machines (an no machine runs OOM). Instance size / complexity can be varied by only including certain families of instances.
%% Each node in a graph is a job, its degree is the job size. Motivation: load balancing for distributed graph algorithms 
%% others? simple load balancing applications in our group?
%% shortcut graphs from route planning ?
%% build processes
%% distributed hash table with a fixed set of fixed-size partitions which must be distributed
%% distribute SAT inprocessing effort using # literals as a metric

%\todo{planted instances? perturbation?}

\subsection{Results.}

\iffalse
\begin{figure}
    \centering
    \begin{minipage}{0.33\textwidth}
        \includegraphics[width=\textwidth]{cdf-bounds-lower.pdf}    
    \end{minipage}%
    \begin{minipage}{0.33\textwidth}
        \includegraphics[width=\textwidth]{cdf-bounds-upper.pdf}    
    \end{minipage}%
    \begin{minipage}{0.33\textwidth}
         \includegraphics[width=\textwidth]{pareto-bounds.pdf}    
    \end{minipage}%
    \caption{
        Quality of different lower (left) and upper (center) bounding techniques;
        point $(x,y)$ indicates that an approach was able to find a bound that deviates from the best found bound by at most $x$ for a share $y$ of all instances.\psfrage{it would be very interesting to see the difference to the optimum here. Thus only look at solvable instances and use the final results?}\todo{DS: We've been discussing this but the BnB results weren't there yet at that point. I'll take a look at this tomorrow}
        On the right, each bounding technique's mean running time is related to its mean\psfrage{This makes little sense as instances with large $U$ will dominate this comparison. Use relative distance and then as geometric or harmonic mean?} (absolute) distance to the best bound.
        Our new techniques are underlined.\psfrage{somewhere explain the color coding? blue is upper bound red lower bound?}
    }
    \label{fig:cdf-bounds}
\end{figure}
\fi

\begin{figure*}[t]
    \centering
    \begin{minipage}{0.39\textwidth}
        \includegraphics[width=\textwidth]{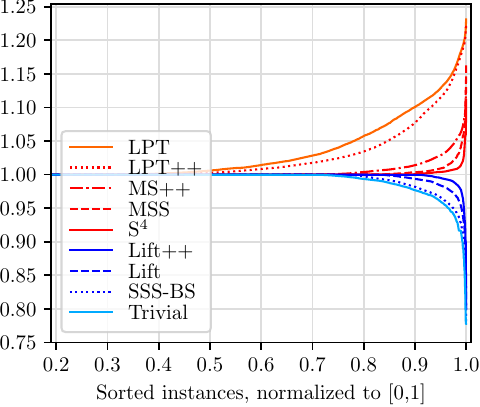}
    \end{minipage}\
    \hspace{0.05\textwidth}
    \begin{minipage}{0.39\textwidth}
        \includegraphics[width=\textwidth]{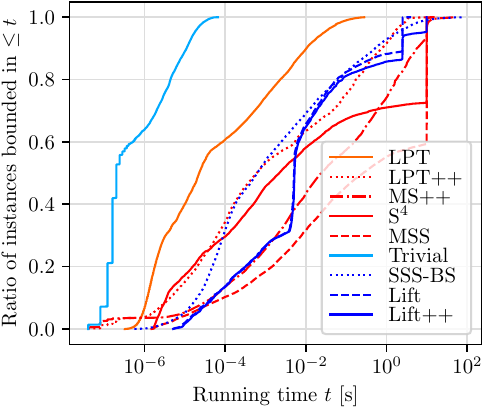}
    \end{minipage}%
    \caption{
        Results of bounding experiments.
        Red denotes upper bounds, blue denotes lower bounds.
        Left: Bounding quality.
        Each curve represents the quotients between a technique's found bound and the corresponding optimal makespan, sorted by their difference to 1 and normalised to the $x$ interval $[0,1]$.
        %For both lower and upper bounds, the optimal curve is a constant line $y=1$.
        Right: Cumulative running time distributions of all bounding techniques (higher is better).
    }
    \label{fig:bounding}
\end{figure*}

First of all, Fig.~\ref{fig:bounding} shows results regarding our bounding techniques (machine A).
In this discussion, we only consider instances for which we know an optimal makespan $C^*$ (16\,277 instances), which allows us to rate each found bound $C$ by how close $C/C^*$ is to 1.
We run each bound with a 10\,s timeout.
In terms of lower bounds, we compare our technique Lift++ to the prior lifting approach (``Lift'') %by Haouari et al.~
\cite{lifting_1}, to SSS-bounding-strengthening (``SSS-BS'')~\cite{lifting_2}, and the trivial bound by Dell'Amico and Martello~\cite{dellamico_martello_1995}.
Lift++ is able to improve on the prior lifting technique's lower bounds on 8\% of all instances, bounds 91.1\% of instances \emph{optimally} (83.4\% for Lift), and always reports the best lower bound among all considered bounding techniques.
In turn, Lift++ incurs a geometric mean slowdown of 21.7\% over Lift.

In terms of upper bounds, we compare our approaches to Multi-Start-Subset-Sum~\cite{lifting_2} and LPT.
MS++ provides an appealing middle ground between basic and high-quality bounding techniques in terms of time-quality tradeoff, whereas S$^4$ is able to further improve upon the high-quality bounds of MSS.
Specifically, with MS++ providing a warm start, S$^4$ bounds 89.1\% of instances optimally (82.6\% for MSS) while also achieving a substantial geometric mean speedup of 20.6\% over MSS.
That being said, we acknowledge that literature on upper bounding techniques is vast whereas we only considered a small selection of techniques.

The best combination of a lower bound and an upper bound---Lift++ and S$^4$---tightly bounds 13\,887 out of 16\,765 instances (82.8\%), hence solving these instances exactly.
%The best singular approach, SSSS with \todo{DS: which LB technique?}, solved 11\,813 instances (70.5\%) exactly.
This reinforces earlier observations that, empirically, many \pcmax{} instances % considered in practice
are easy to solve exactly with good bounding techniques~\cite{lifting_2}.
Only 8--10\% of CNF, graph, and planted instances %turned out to be
are non-trivial in this sense, whereas for running time instances this ratio varied drastically, ranging from 0\% for SAT solving to 50\% for taxi-sharing queries.

\begin{figure}
    \centering
    \includegraphics[width=0.7\columnwidth]{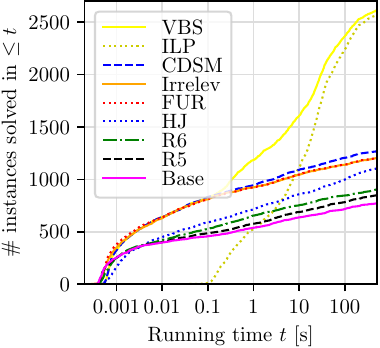}
    \caption{
        %Left: 
        Cumulative running times for \bnb{} and ILP approaches, disregarding instances solved exactly by bounding.
        %Right: The same distribution where each instance from domain family $\mathcal{F}$ is weighed by $1/|\mathcal{F}|$.
        Note the log.\ scale in $x$ direction.
        VBS (``Virtual Best Solver'') depicts a virtual running time oracle combining CDSM and ILP.
    }
    \label{fig:cdf-cdsm-configs}
\end{figure}

\begin{table}
    \centering
    \small
    \setlength{\tabcolsep}{2.5pt}
    \begin{tabular}{lrrrrr}
        \toprule
                           &  R5    & R6     & FUR    & Irr.   & CDSM \\
        \midrule
        Solved ratio       &  1.10  & 1.17   & 1.56   & 1.56   & 1.64 \\
        Speedup            &  1.04  & 1.66   & 14.21  & 13.36  & 11.91 \\
        Expl.\ node ratio  &  1.10  & 2.59   & 59.92  & 59.52  & 89.81 \\
        \bottomrule
    \end{tabular}\\[3mm]
    \textit{Base: 771 instances solved; geom.\ mean time 0.045\,s;\\73\,778 geom.\ mean explored nodes}
    \caption{Relative improvements over ``Base'' w.r.t.\ solved instances, running times, and explored nodes.
        %For each ratio, a number $x>1$ denotes an \textit{improvement} over the respective property by factor $x$.
        %Speedups and explored node ratios are w.r.t.\ commonly solved instances only.
        %Base solved 771 instances at a geom.\ mean running time of 0.045\,s exploring a geom.\ mean of 73\,778 nodes.
        }
    \label{tab:cdf-cdsm-configs}
\end{table}

% Figures are rasterized now due to their rendering complexity.
\begin{figure*}[t]
    \centering
    \begin{minipage}{0.3\textwidth}
        \includegraphics[width=\textwidth]{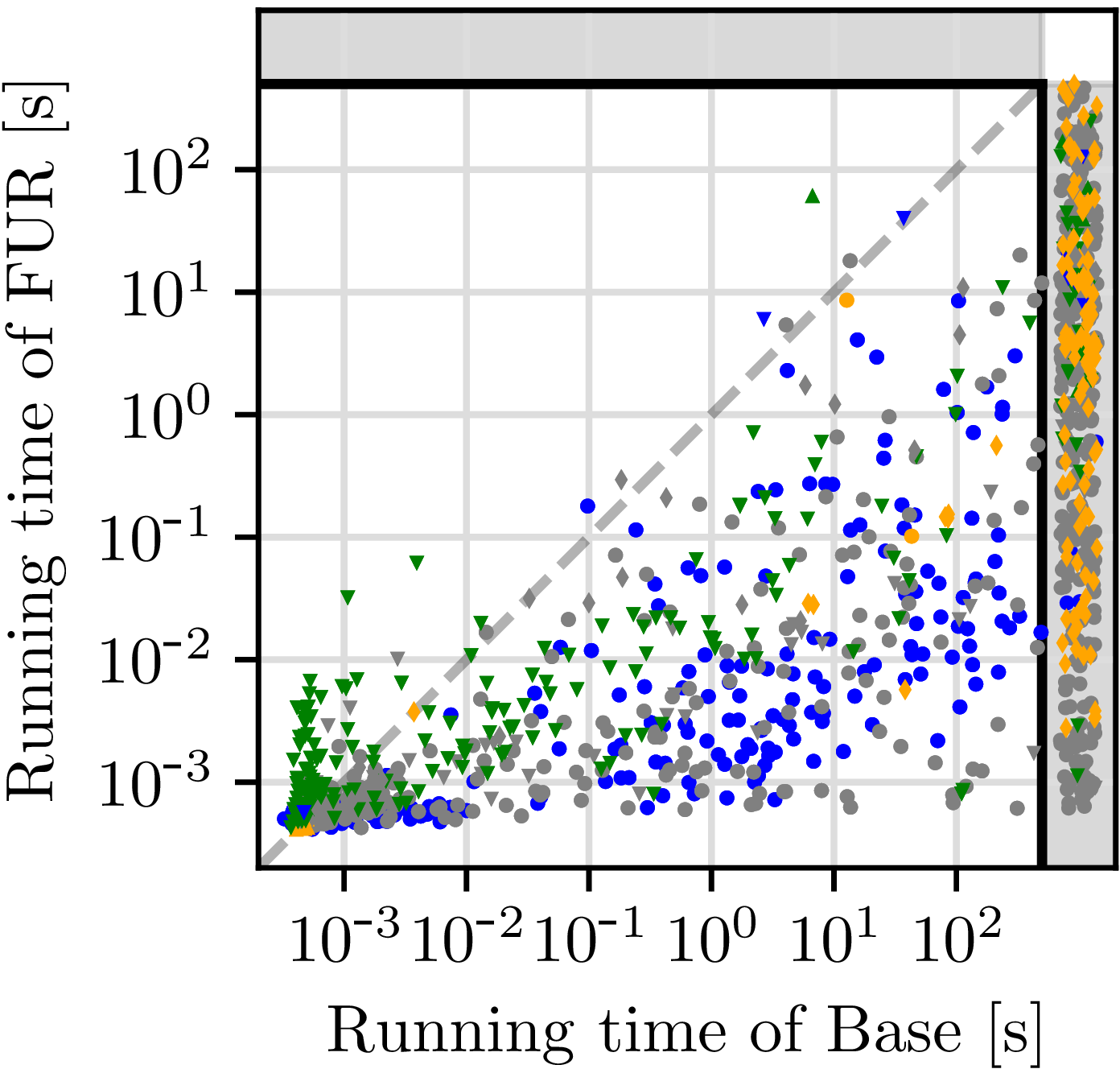}
    \end{minipage}\ \ %
    \begin{minipage}{0.3\textwidth}
        \includegraphics[width=\textwidth]{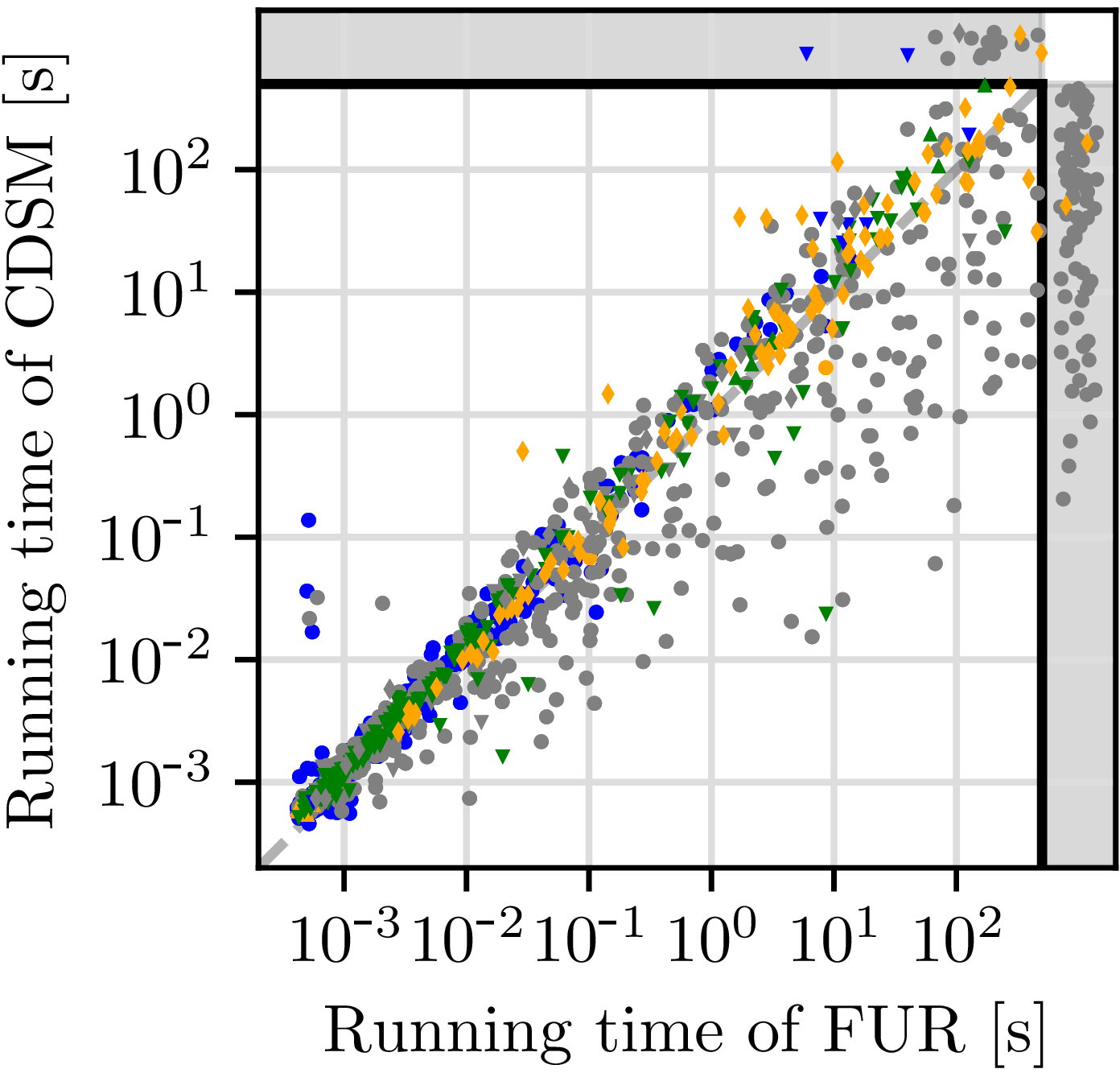}
    \end{minipage}\ \ %
    \begin{minipage}{0.3\textwidth}
        \includegraphics[width=\textwidth]{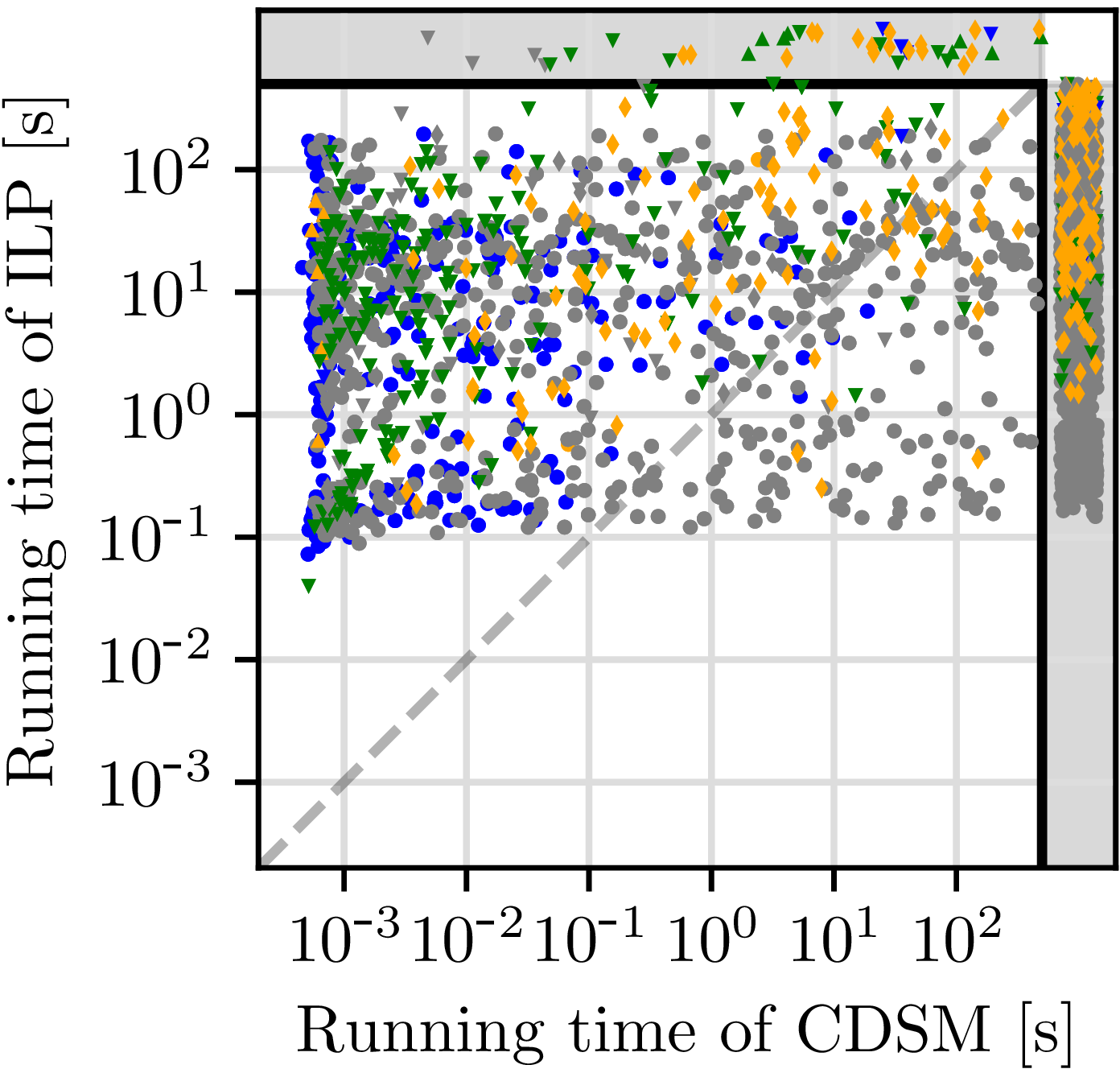}
    \end{minipage}\\[1mm]
    \includegraphics[width=\textwidth]{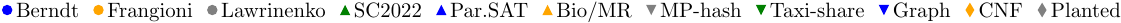}
    \caption{
        Direct comparison of running times of Base vs.\ FUR (left), FUR vs.\ CDSM (center), and CDSM vs.\ the ILP approach (right).
        Note the logarithmic scale in both directions.
    }
    \label{fig:1v1}
\end{figure*}

\begin{table*}[t]
    \centering
    \footnotesize
    \setlength{\tabcolsep}{2.5pt}
    \begin{tabular}{lrrrrrrrrrrrr}
        \toprule
                           &        &       & \multicolumn{8}{c}{\# solved exactly} & \multicolumn{2}{c}{VBS} \\
        \cmidrule(lr){4-11}\cmidrule(lr){12-13}
        Family             &  \#    & \#$^*$ & \ \ Base &  R5    & R6     & FUR    & Irr. & CDSM & HJ   & ILP & \ \ two & all \\ 
        \midrule
        Berndt      &  3060   &  316   &  305        &  304             &  309         &  313       &  313               &  313        &  288   &  313           &  313                         &  316      \\  
        Frangioni   &  780    &  38    &  2          &  2               &  2           &  2         &  2                 &  2          &  16    &  6             &  6                           &  18       \\  
        Lawrinenko  &  3500   &  1569  &  205        &  204             &  246         &  473       &  472               &  539        &  366   &  1558          &  1558                        &  1558     \\  
        SC2022      &  307    &  11    &  1          &  3               &  3           &  8         &  8                 &  8          &  11    &  0             &  8                           &  11       \\  
        RAxML/MR    &  53     &  7     &  7          &  7               &  7           &  7         &  7                 &  7          &  7     &  7             &  7                           &  7        \\  
        MP-hash     &  335    &  125   &  63         &  63              &  67          &  77        &  77                &  79         &  53    &  115           &  119                         &  119      \\  
        Taxi-share  &  746    &  403   &  160        &  160             &  164         &  189       &  189               &  189        &  183   &  233           &  242                         &  257      \\  
        Graph       &  292    &  31    &  3          &  4               &  5           &  8         &  8                 &  6          &  20    &  6             &  10                          &  25       \\  
        CNF         &  6265   &  589   &  8          &  84              &  83          &  104       &  105               &  104        &  121   &  284           &  301                         &  395      \\  
        Planted     &  1064   &  99    &  17         &  17              &  17          &  21        &  21                &  20         &  37    &  49            &  49                          &  69       \\  
        \midrule
        Total       &  16765  &  3188  &  771        &  848             &  903         &  1202      &  1202              &  1267       &  1102  &  2571          &  2613                        &  2775     \\
        \bottomrule
    \end{tabular}

    \vspace{5pt}
    \caption{Number of exactly solved instances by benchmark family. ``\#$^*$'' signifies the number of instances that were not solved by our bounding heuristic, i.e., the number of instances the \bnb{} and ILP approaches were run on.
    The final two columns show a virtual running time oracle (VBS, Virtual Best Solver) of \{CDSM, ILP\} (``two'') and of all eight displayed runs (``all'').
    }
    \label{tab:solved-instances}
\end{table*}

We compare our \bnb{} algorithm
%For our exact solving scheme, we compare our performance
to two competitors from literature:\footnote{
    As an additional point of reference, CDSM with Lift++/S$^4$ solves 99.4\% of all Berndt instances in under 11\,s whereas Berndt et al.'s EPTAS reportedly takes over an hour for the majority of instances~\cite{berndt2022load}.
} the ILP encoding by Mrad and Souayah~\cite{8248744} running on the Gurobi optimizer~\cite{gurobi}, and the \bnb{} scheme by Haouari and Jemmali~\cite{lifting_2} (reimplemented).
For each problem instance, we apply all of the above bounding techniques\footnote{
    %These tests were executed with hyper-threading disabled, causing slightly better performance.
    We ran these tests on machine B, leading to deviating bounds from the bounding experiments.
} for up to 10\,s and then run \bnb{} or ILP \textit{only if the found bounds are not yet tight} for up to 500\,s.
Fig.~\ref{fig:cdf-cdsm-configs} summarizes the performance on these instances.
%shows a general overview on the number of exactly solved instances.
%w.r.t.\ a certain time limit.
%Since our eleven instance families have very different sizes (see Tab.~\ref{tab:solved-instances}), we also provide a normalized illustration where the instances from each family sum up to 1 (Fig.~\ref{fig:cdf-cdsm-configs} right).

\ifappendix\else
    \begin{figure*}[t!]
        \centering
        \includegraphics[width=0.9\textwidth]{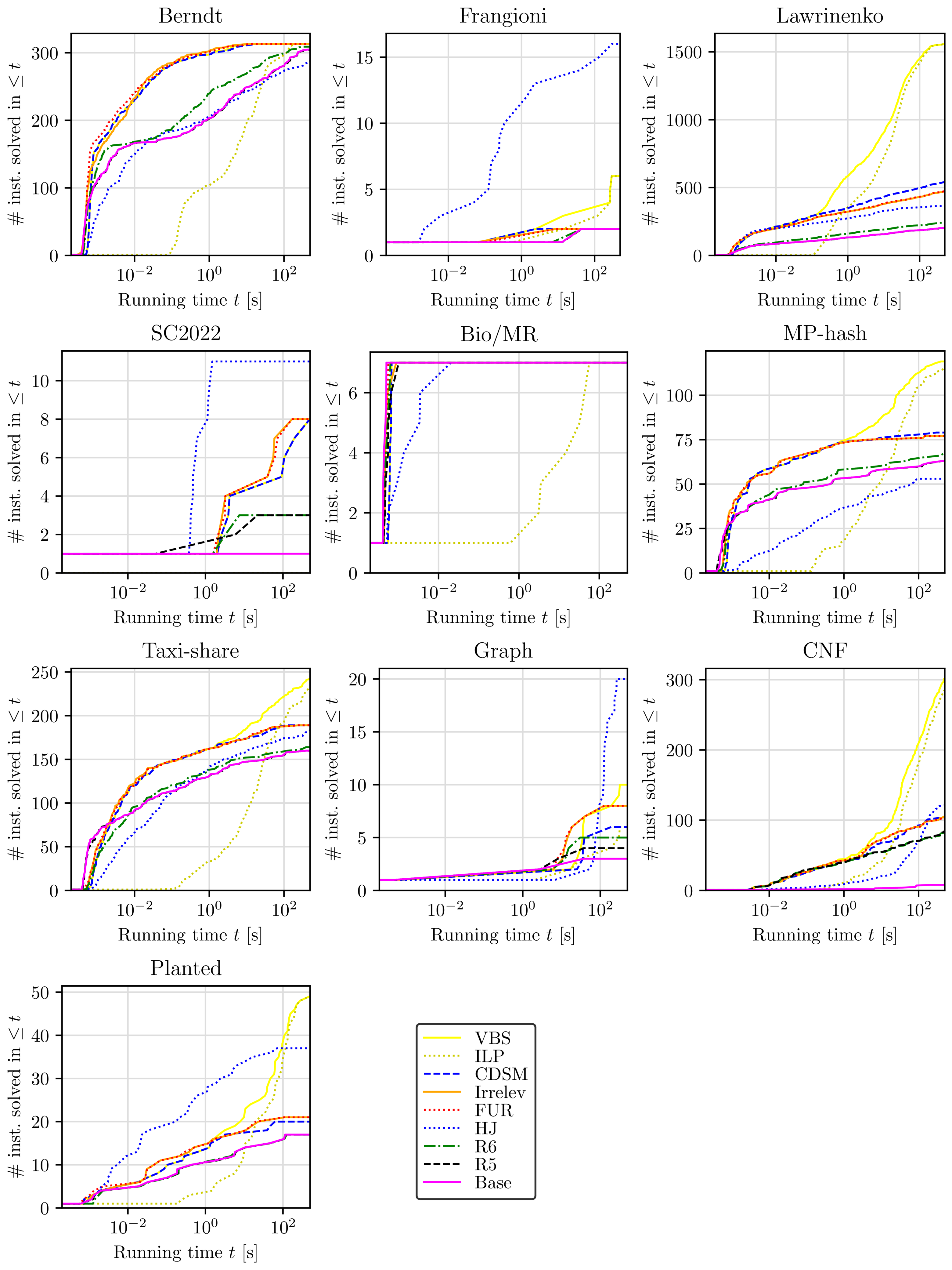}
        \caption{Cumulative running times split by family (without Par.SAT since no instances were solved).
            VBS (``Virtual Best Solver'') depicts a virtual running time oracle combining CDSM and ILP.}
        \label{fig:cdfs-by-domain}
    \end{figure*}
\fi

We considered the following increments of our \bnb{} algorithm: a base version which corresponds to Dell'Amico and Martello's algorithm~\cite{dellamico_martello_1995}; adding R\ref{pruning_rule:equal_remaining_jobs}; adding R\ref{pruning_rule:inter} (and thus maintaining the RET); adding the FUR a.k.a.\ R\ref{pruning_rule:fur}; adding the application of Theorem~\ref{theorem:fur_upper_bound} (and thus maintaining irrelevance index \algirrelevanceindex); and, finally, running full CDSM (and thus maintaining the \algstatetable).
Let us briefly discuss the performance of each increment. %following our base version.
First, R\ref{pruning_rule:equal_remaining_jobs} brings a significant improvement on the CNF instances, which have large numbers of equally sized small jobs (see Tab.~\ref{tab:benchmark-families-jobsize-distributions}); as soon as all jobs of non-final size have been placed, the problem can be solved immediately. %, which results in the clearly visible steep rise of the CDF curves in Fig.~\ref{fig:cdf-cdsm-configs} between 1 and 10\,s of running time.
Note that none of the prior benchmarks from literature allow to observe this behaviour.
R\ref{pruning_rule:inter} improves performance on several benchmark families; the number of explored nodes is cut in half (geom.\ mean factor 2.17).
By far the most substantial leap in performance is achieved by %R\ref{pruning_rule:equal_remaining_jobs} and when adding
the FUR.
Across the 901 instances solved by both the FUR and the earlier R\ref{pruning_rule:inter} configuration, FUR results in a geometric mean speedup of 8.4 and the number of explored nodes drops by a mean factor of 20.8.
Comparing FUR to the base configuration, Fig.~\ref{fig:1v1} (left) illustrates that running times often improve by several orders of magnitude.
The increment enabling irrelevance rules is the only one without any clear merit.
%\todo{NM: can we say that explored nodes is slightly reduced but overshadowed by running time overhead?}
%\todo{DS: No; both get worse. Speedup 0.94; explored nodes factor 0.99
%NM: I don't understand how this is even possible if our algorithm is deterministic (since ST is not yet involved).}
Lastly, CDSM improves performance drastically on Lawrinenko's instances---relatively small instances designed to be difficult---but incurs noticeable overhead on most other instances (Fig.~\ref{fig:1v1} center) due to maintaining the \algstatetable{}.
As such, the mean speedup from FUR to CDSM on commonly solved instances is very small (1.03).
Still, CDSM is able to solve 65 additional instances and reduce the number of explored nodes by a mean factor of 2.06.
Put together, our improvements to the algorithm of Dell'Amico and Martello~\cite{dellamico_martello_1995} increase the number of solved instances by 64\%, reduce explored nodes by a factor of 89.8, and result in a mean speedup of 11.9 on instances solved by both.

% (see above)
%\psfrage{somewhere discuss the curious shape of most curves having a small steep leg and then levelling off. Intpretation: certain heuristic only begin paying off at difficult instances? some stuff is actually activated late?}\todo{MA: This is due to the CNF instances, with many small jobs, R5 can easily solve most of them, but requires a bit of overhead. All of those instances are roughly of similar size, so they all finish basically at the same time, casuing that weird shape.}

\begin{table}[t!]
    \centering
    \footnotesize
    \setlength{\tabcolsep}{3pt}
    \begin{tabular}{lrrrrrrr}
        \toprule
        & \multicolumn{7}{c}{Job size} \\
        \cmidrule(rr){2-8}
        Family      &  min  &  p25  &  med  &  geom   &  arith   &  p75    &  max     \\
        \midrule
        Berndt      &  1    &  84   &  187  &  167.3  &  303.9   &  514   &  1\,000     \\
        Frangioni   &  1    &  94   &  892  &  583.9  &  2663.3  &  4873  &  10\,000   \\
        Lawrinenko  &  1    &  64   &  92   &  110.7  &  187.7   &  198   &  1\,658    \\
        SC2022      &  1    &  10   &  129  &  127.8  &  2494.6  &  1288  &  50\,004   \\
        Par.SAT     &  2    &  46   &  160  &  135.3  &  366.0   &  428   &  3\,097    \\
        Bio/MR      &  2    &  73   &  233  &  248.1  &  2505.5  &  739   &  526\,875  \\
        MP-hash     &  340  &  439  &  477  &  513.8  &  551.8   &  535   &  16\,365   \\
        Taxi-share  &  111  &  421  &  691  &  722.0  &  987.0   &  1188  &  60\,681   \\
        Graph       &  1    &  5    &  7    &  8.8    &  27.8    &  10    &  26\,068   \\
        CNF         &  2    &  2    &  3    &  2.9    &  3.3     &  3     &  3\,222    \\
        Planted     &  1    &  32   &  97   &  92.6   &  268.9   &  298   &  3\,000    \\
        \bottomrule
    \end{tabular}
    \caption{Job size distributions for each benchmark family; featuring the minimum, median, geometric and arithmetic mean, maximum, and the 25th and 75th percentile.}
    \label{tab:benchmark-families-jobsize-distributions}
\end{table}

As Fig.~\ref{fig:cdf-cdsm-configs} indicates, the HJ algorithm is
%outperformed by all increments of our algorithm which feature the FUR.
also outperformed by our approach.
Still, the solved instances are orthogonal to a degree.
%Still, the algorithms appear to behave orthogonally to a degree.
Tab.~\ref{tab:solved-instances} and Fig.~\ref{fig:cdfs-by-domain} \ifappendix(App.)\fi\ show that HJ performs poorly on Berndt, Lawrinenko, and MP-hash instances, which share a moderate number of processors (mostly $\leq 100$) and small $n/m$ (mostly $\leq 10$).
%, and modest job sizes (Fig.~\ref{fig:benchmark-families-m-distributions}, Tab.~\ref{tab:benchmark-families-jobsize-distributions}).
By contrast, HJ outperforms all others on Frangioni, SC2022, and Graph instances, which share especially long %average
optimal makespans (%tens to hundreds of thousands of time units;
see \ifappendix App.\ \fi Tab.~\ref{tab:benchmark-families-optimum-distributions}).
Interestingly, HJ also outperforms CDSM on planted instances, which have makespans $\leq 3000+\varepsilon$.
We suspect that HJ's different branching heuristics are
especially useful to reconstruct the planted solutions; further analyses are needed to investigate this effect.
%\todo{DS: I'll see if I can analyze these planted instances some more}
%beneficial to find the ``hidden'' optimal solutions in these particular instances.

\ifappendix\else
    \begin{table*}[t]
        \centering
        \footnotesize
        \setlength{\tabcolsep}{3.5pt}
        \begin{tabular}{lrrrrrrrrr}
            \toprule
            & \multicolumn{9}{c}{Found optimal makespan} \\
            \cmidrule(rr){2-10}
            Family      &  min   &  p10   &  p25    &  med     &  geom           &  arith          &  p75     &  p90      &  max      \\
            \midrule
            Berndt      &  13    &  59    &  171    &  453     &  446.7     &  1563.6    &  1213    &  4436     &  24926    \\
            Frangioni   &  86    &  471   &  1891   &  9861    &  11506.1   &  119130.8  &  93938   &  376316   &  1882118  \\
            Lawrinenko  &  90    &  128   &  158    &  229     &  280.9     &  406.7     &  322     &  1086     &  2477     \\
            SC2022      &  1667  &  9969  &  55546  &  181727  &  182859.6  &  848179.7  &  640124  &  3228482  &  6892847  \\
            Par.SAT     &  2903  &  3712  &  5206   &  11587   &  11110.3   &  15040.8   &  23017   &  30231    &  49036    \\
            Bio/MR      &  135   &  228   &  419    &  670     &  3461.7    &  195810.3  &  54644   &  281311   &  3033958  \\
            MP-hash     &  930   &  1307  &  1595   &  2175    &  2242.7    &  2595.1    &  2632    &  5252     &  16802    \\
            Taxi-share  &  801   &  1363  &  2064   &  4569    &  5096.1    &  9666.9    &  10847   &  26784    &  77495    \\
            Graph       &  5     &  20    &  62     &  394     &  544.6     &  10428.3   &  3700    &  27747    &  333332   \\
            CNF         &  5     &  33    &  117    &  840     &  826.5     &  6601.6    &  5641    &  20727    &  201695   \\
            Planted     &  100   &  101   &  301    &  1000    &  663.6     &  1198.6    &  3000    &  3001     &  3002     \\
            \bottomrule
        \end{tabular}
        \caption{Distributions over optimal makespans for each family; featuring the minimum, median, geometric and arithmetic mean, maximum, and the 10th, 25th, 75th, and 90th percentile.}
        \label{tab:benchmark-families-optimum-distributions}
    \end{table*}
\fi

%The above internal comparison showed that our approach is able to \textit{double} the number of solved instances over the prior state of the art in \bnb{}-based \pcmax{} scheduling~\cite{dellamico_martello_1995}.
Next, we compare these \bnb{} results to the ILP approach~\cite{8248744}, which we consider to reflect the state of the art in \pcmax{}.
As Fig.~\ref{fig:cdf-cdsm-configs} and Fig.~\ref{fig:1v1} (right) show, the ILP approach requires a significant upfront investment to generate the ILP encoding, resulting in more than two orders of magnitude of overhead over \bnb{} for the easiest of instances.
This investment does pay off after 10\,s, at which point the ILP approach surpasses all \bnb{} approaches.
Given the full 500\,s, ILP solves substantially more instances (2571 vs.\ 1267 for CDSM), especially on the benchmark set where it was originally evaluated (Lawrinenko)
as well as MP-hash, Taxi-share, CNF, and planted instances.
These families share a modest %average
optimal makespan (few thousand).
By contrast, \bnb{} %tend to
performs better on families with large optimal makespans, such as SC2022, Frangioni, and Graph.
%These instances feature very large values of $n$, where the ILP encoding gets too large.
%As such, we conclude that the former ILP approach and our \bnb{} approach are to a degree orthogonal and excel on different kinds of instances.
For example, CDSM found more solutions with a makespan $\geq 10\,000$ than ILP (87 vs.\ 71).
%As a consequence,
This indicates that \bnb{} may also
be the more favorable exact approach
%get increasingly favorable
for higher \textit{resolutions} of job sizes.
A virtual oracle which picks the approach with lower running time from CDSM and ILP for each instance solves 42 additional instances compared to ILP and, more importantly, results in a mean speedup of 13 over ILP.
An oracle based on ILP and HJ even solves 178 additional instances but only with a speedup of 6.
%\todo{$\leftarrow$ Include this info? :SD -- NM: It is mostly inferrable from Table 1, so I think it is okay to omit it}
As such, running our approach for just one second before proceeding with the more heavyweight ILP approach is already an appealing option to boost an exact \pcmax{} scheduler's performance.
%Our observations also confirm that considering new, application-based benchmark sets can be useful to expose behavior which prior synthetic benchmark instances do not sufficiently capture.
%\todo{NM: last sentence is a bit out of context. Also, a very similar statement is in the conclusion which is directly afterwards, so perhaps we don't need it.}

%\todo{DS: say something about memory usage, especially CDSM vs ILP}

\section{Conclusion}\label{sec:conclusion}

%\todo{DS: I usually write my conclusions in the past tense. I'm not against writing it like below though. NM: did a new round of smaller adjustments}

We present new pruning rules for the \pcmax{} \emph{decision problem}
and integrate these rules into an efficient \bnb{} scheme for optimally solving the \pcmax{} \emph{optimization problem}.
Together with a lower bounding technique which clearly outperforms previous methods and multiple improvements to upper bounding techniques, this allows to solve some instances with thousands of machines and jobs within seconds.
%\todo{DS: New sentence:}
Our experimental analyses are supported by new benchmark sets crafted from diverse application data, which exhibit characteristics
%are indeed able to shed light on some aspects
that synthetic benchmarks from prior \pcmax{} literature fail to capture.
%\todo{DS: dropped mem usage since it's actually worse than we thought :)}
Although a cutting-edge commercial ILP solver still performs substantially better if given sufficient time,
our prototypical implementation is the faster approach and appears to scale better to huge makespans and, therefore, fine discretisations of job sizes. %, solving many instances in sub-second time and with much less memory usage.
All in all, we clearly advanced the state of the art for \bnb{}-based scheduling. %, as demonstrated by our detailed evaluation with a large and diverse set of newly created instances from real-world data.

\iffalse
Compared to a cutting-edge commercial ILP solver,
% executed on a start-of-the-art \pcmax{} encoding,
our prototypical implementation %cannot yet solve a similar number of solved instances but 
is the faster and more lightweight approach if only few seconds of running time or little memory are available, whereas the ILP approach still performs substantially better if given sufficient resources.
\fi
%Our approach is complementary to the ILP approach, which 
%While the ILP approach is the best available technique for smaller instances and long running times,
%Our approach is also complementary to an efficient ILP encoding, which is the best available technique for smaller instances,
%our \bnb{} algorithm is uncontested for low running time limits and excels on some instance families where the ILP approach is comparatively weak.\todo{DS: doesn't really hold any longer}

%In addition, 
Our work highlights multiple options for further improvements.
Our efficient base case for a single remaining job duration has proven impactful,
which makes it promising to investigate techniques for extending it to multiple job durations.
%(cf.\ App.~\ref{appendix:two_durations})
Both RET and CDSM are powerful techniques that we believe to bear further potential %could potentially be exploited 
for even more pruning.
In addition, we aim to further optimize our implementation %, while reasonably efficient, is neither highly engineered %for optimal running time 
%and it also does not support parallelism.
%nor parallel, which we intend to address.
and explore possible parallelisations.
Lastly, beyond algorithmic improvements, it might also prove insightful to investigate the differences between instance families and their effect on solvers in more detail.

\vspace{.2cm}
\paragraph{Acknowledgements}
This project has received funding from the pilot program
Core–Informatics of the Helmholtz Association (HGF).

This project has received funding from the European Research Council (ERC) under the European Union’s Horizon 2020 research and innovation programme (grant agreement No. 882500).

\begin{center}
    \includegraphics[width=0.5\linewidth]{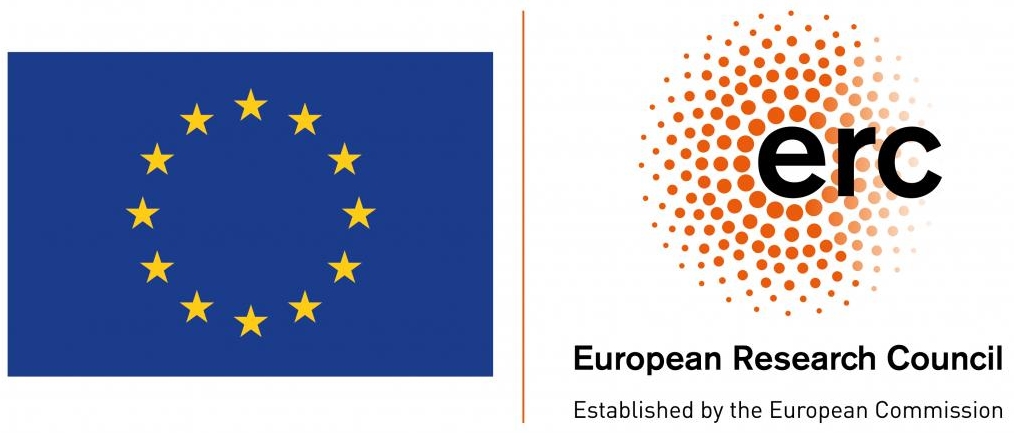}
\end{center}

\printbibliography

\ifappendix
\clearpage

\appendix

\section{Lifting Procedure}\label{appendix:lifting}
The following lifting procedure is presented by Haouari et al.~\cite{lifting_1}.
Haouari and Gharbi~\cite{lifting_minus_1} prove that in any \pcmax{} instance, there are $\alpha$ processors that process at least  
\[
    \lambda_\alpha(n,m) = \alpha \lfloor n/m \rfloor + \min(\alpha, n- \lfloor n/m \rfloor m)
\]
jobs.
We can then define a sub-instance $I'$ on the $\beta \leq n$ largest jobs.
For this instance, and for a suitable $\alpha$, we create a further sub-instance which we refer to as $I_{\beta}^{\alpha}$ by taking the smallest $\lambda_\alpha(\beta, m)$ jobs and $\alpha$ processors.
Lower bounds on $I_{\beta}^{\alpha}$ are also valid bounds on our original \pcmax{} instance.
With that, for a lower bound function $L$, we can improve it by calculating
\[
\overline{L}(I) = \max_{1 \leq \alpha \leq m} \{ \max_{\alpha < \beta \leq n } L(I_\beta^\alpha)\}.
\]

They then show that the only values of $\beta$ that have to be considered are the values such that $\beta = \gamma m + \alpha$ for $\gamma = 1, \ldots, \left\lfloor \frac{n-\alpha}{m}\right\rfloor$.
With that, the number of generated sub-instances is in $O(n)$.
By using other bounding techniques, we can calculate good lower bounds for the generated sub-instances.

%\section{The Trivial Lower Bound}\label{appendix:trivial_bound}
%The simplest lower bound for a \pcmax{} instance \pcmaxinstance{} is defined as follows~\cite{dellamico_martello_1995, 10.1287/mnsc.6.1.1}:
%\[
%L_{TV} = \max \{w_1, \ w_{m} + w_{m+1}, \ \left\lceil \frac{1}{m}\sum_{w_i \in W}{w_i} \right\rceil\}.
%\]

\section{Dynamic Programming for Two Remaining Job Durations}
\label{appendix:two_durations}

We can extend the idea from R\ref{pruning_rule:equal_remaining_jobs} to formulate pruning rules for two remaining job durations,
thereby improving the base case even further.
In the following, we consider a decision problem instance $(W, m, U)$ and a valid partial assignment $A$ at decision level $\decisionlvl$.
Note that none of these results is implemented yet, which is a topic for future work.
%In addition, we always assume $w_{\decisionlvl + 1} = \cdots = w_r$ and $w_{r+1} = \cdots = w_n$ for some $r$ with $\decisionlvl < r < n$.
\psfrage{There are natural generalizations for a pruning rule: round down $w_{\decisionlvl + 1}\ldots w_r$ to $w_r$ and $w_{r+1}\ldots w_n$ to $w_n$. If the resulting 2-weight instance is infeasible then also the actual instance is infeasible. Of course, this opens up a can of tuning parameter worms but it might be useful. For example, look only at the value for $r$ maximizing the rounded down weight. Only if that weight is more than $\alpha$ time the actual weight of the remaining jobs, actually do the dynamic programming. This parameter could also be tuned adaptively by monitoring the fraction of the time spent on dynamic programming. If this exceeds some threshold, increase $\alpha$.}
%\todo{[NM: Yes this seems like it could be quite useful, but indeed requires a lot of tuning]}\psfrage{Once more, this also works for the optimization problem.}

%\subsection{Dynamic Programming for Two Job Durations}

We can solve the general case of two remaining job durations in polynomial time with a dynamic programming formulation.
Assume that there are $n_a$ remaining jobs with duration $a$ and $n_b$ remaining jobs with duration $b$
(and no other remaining jobs).
Our dynamic program maximizes the number of jobs with duration $b$ which can still be assigned after all jobs with duration $a$ are assigned.
Note, this works for both $a > b$ or $b > a$.
%\psfrage{save some indices by renaming $w_r$ and $w_n$? $w$ and $w'$? $a$ and $b$? Do things simplify if we make clear which of the jobs is larger? MA: For implementation, it would be better to swap the roles of the two job sizes if we have way more of the second one than the first one, there is no reason why we have to build the DP table with the larger size.}
%\todo{[NM: see new version]}
%\psfrage{Its confusing that for naming the job sizes you use $r$ and $n$ while for the dynamic programming you start at one. Its more elegant to transform to a new problem starting at one?}

\begin{theorem}\label{theorem:two_duration_dp}
    Let $T\colon [m] \times \mathbb{N}_0 \rightarrow \mathbb{N}_0 \cup \{ -\infty \}$ be a function such that $T(x, k)$ describes
    the number of jobs with duration $b$ which are assignable to the first $x$ processors under the assumption that $k$ jobs with duration $a$ are already assigned to these processors.
    If we define
    \[
        t(x, k) :=
        \begin{cases}
            -\infty & \text{if } C_x^{\decisionlvl} + k \cdot a > U \\
            \left\lfloor \frac{U - C_x^{\decisionlvl} - k \cdot a}{b} \right\rfloor
            & \text{if } C_x^{\decisionlvl} + k \cdot a \le U 
        \end{cases}
    \]
    and define $T$ recursively as
    \[
        T(x, k) :=
        \begin{cases}
            t(x, k) & \text{if } x = 1 \\
            \max_{0 \le i \le k} T(x - 1, k - i) + t(x, i) & \text{if } x > 1
        \end{cases},
    \]
    then $T(m, n_a)$ is the maximal achievable number of still assignable jobs with duration $b$
    after all jobs with duration $a$ are assigned.
\end{theorem}

\begin{proof}
    It is sufficient to show that $T(x, k)$ is actually
    the maximal possible number of assignable jobs with duration $b$ when assigning $k$ jobs with duration $a$ to the first $x$ processors,
    or $-\infty$ if no valid assignment exists.
    Clearly, this is the case for $x=1$.
    For larger $x$, this follows by induction
    since the maximum covers all possible cases of splitting the available jobs between the current processor and the processors with smaller index.
\end{proof}

Since the table has at most $m \cdot n_a$ entries and computing an entry needs at most $n_a$ lookups,
this theorem allows to decide in time $O(m \cdot n_a^2)$ whether a valid completion of $A$ exists.
For an actual implementation, some further improvements are possible.
First, for the computation of the minimum, it suffices to consider values of $i$ where both terms are finite
(and returning $-\infty$ if no such $i$ exists).
The maximum index where the entry for processor $p_x$ is finite is $\sum_{y=1}^x \max \{ k \mid t(x, k) > -\infty \}$,
thus these upper bounds can be computed via a prefix sum over the processors.
Second, even the remaining entries are not all required
(e.g., we need only the entry with index $n_a$ from the last column).
This allows to compute a lower bound for $i$ in a similar way to the upper bound, but starting from the last processor.

\psfrage{Pruning Rule 5. There seem to be interesting generalizations?\\ If $\sum_{x=1}^m \left\lfloor \frac{U - C_x^{\decisionlvl}}{\min_{j=\ell+1}^{n}w_j} \right\rfloor > n - \decisionlvl$ then there is no feasible solution regardless whether unassigned jobs have equal duration and regardless of the ordering of the jobs. This rule also works for the optimization problem with the lower bound?\\ Similarly, if $\sum_{x=1}^m \left\lfloor \frac{U - C_x^{\decisionlvl}}{\max_{j=\ell+1}^{n}w_j} \right\rfloor \le n - \decisionlvl$ then there is a feasible solution. The latter is perhaps not relevant as the DFS will dive down to a solution anyway.}

\psfrage{Theorem 2. Can this be generalized to a pruning rule that can be used at any decision level? sth like \[
        \sum_{i=\ell+1}^{n-1} w_i < mU-\sum_{i=1}^m C_x^{\ell} - ??? + ???,
    \] or is that never helpful?}\todo{NM: would be equivalent, I think}

\todo{NM: Concretely, I think it might be possible to strengthen R\ref{pruning_rule:interchangeable_jobs} a bit.}

\todo{NM: Future work: investigate branching heuristic that uses RET to try minimize remaining space on a processor}
\todo{NM: Future work: Better infeasibility check by using the RET? (necessarily unused space even for not yet filled processors)}

%\clearpage

\section{Supplementary Material}
\label{appendix:supplementary}

On the next pages we provide some additional figures and tables.

\begin{table*}[h!]
    \centering
    \footnotesize
    \setlength{\tabcolsep}{3.5pt}
    \begin{tabular}{lrrrrrrrrr}
        \toprule
        & \multicolumn{9}{c}{Found optimal makespan} \\
        \cmidrule(rr){2-10}
        Family      &  min   &  p10   &  p25    &  med     &  geom           &  arith          &  p75     &  p90      &  max      \\
        \midrule
        Berndt      &  13    &  59    &  171    &  453     &  446.7     &  1563.6    &  1213    &  4436     &  24926    \\
        Frangioni   &  86    &  471   &  1891   &  9861    &  11506.1   &  119130.8  &  93938   &  376316   &  1882118  \\
        Lawrinenko  &  90    &  128   &  158    &  229     &  280.9     &  406.7     &  322     &  1086     &  2477     \\
        SC2022      &  1667  &  9969  &  55546  &  181727  &  182859.6  &  848179.7  &  640124  &  3228482  &  6892847  \\
        Par.SAT     &  2903  &  3712  &  5206   &  11587   &  11110.3   &  15040.8   &  23017   &  30231    &  49036    \\
        Bio/MR      &  135   &  228   &  419    &  670     &  3461.7    &  195810.3  &  54644   &  281311   &  3033958  \\
        MP-hash     &  930   &  1307  &  1595   &  2175    &  2242.7    &  2595.1    &  2632    &  5252     &  16802    \\
        Taxi-share  &  801   &  1363  &  2064   &  4569    &  5096.1    &  9666.9    &  10847   &  26784    &  77495    \\
        Graph       &  5     &  20    &  62     &  394     &  544.6     &  10428.3   &  3700    &  27747    &  333332   \\
        CNF         &  5     &  33    &  117    &  840     &  826.5     &  6601.6    &  5641    &  20727    &  201695   \\
        Planted     &  100   &  101   &  301    &  1000    &  663.6     &  1198.6    &  3000    &  3001     &  3002     \\
        \bottomrule
    \end{tabular}
    \caption{Distributions over optimal makespans for each family; featuring the minimum, median, geometric and arithmetic mean, maximum, and the 10th, 25th, 75th, and 90th percentile.}
    \label{tab:benchmark-families-optimum-distributions}
\end{table*}

\begin{figure*}[h!]
    \centering
    \includegraphics[width=0.9\textwidth]{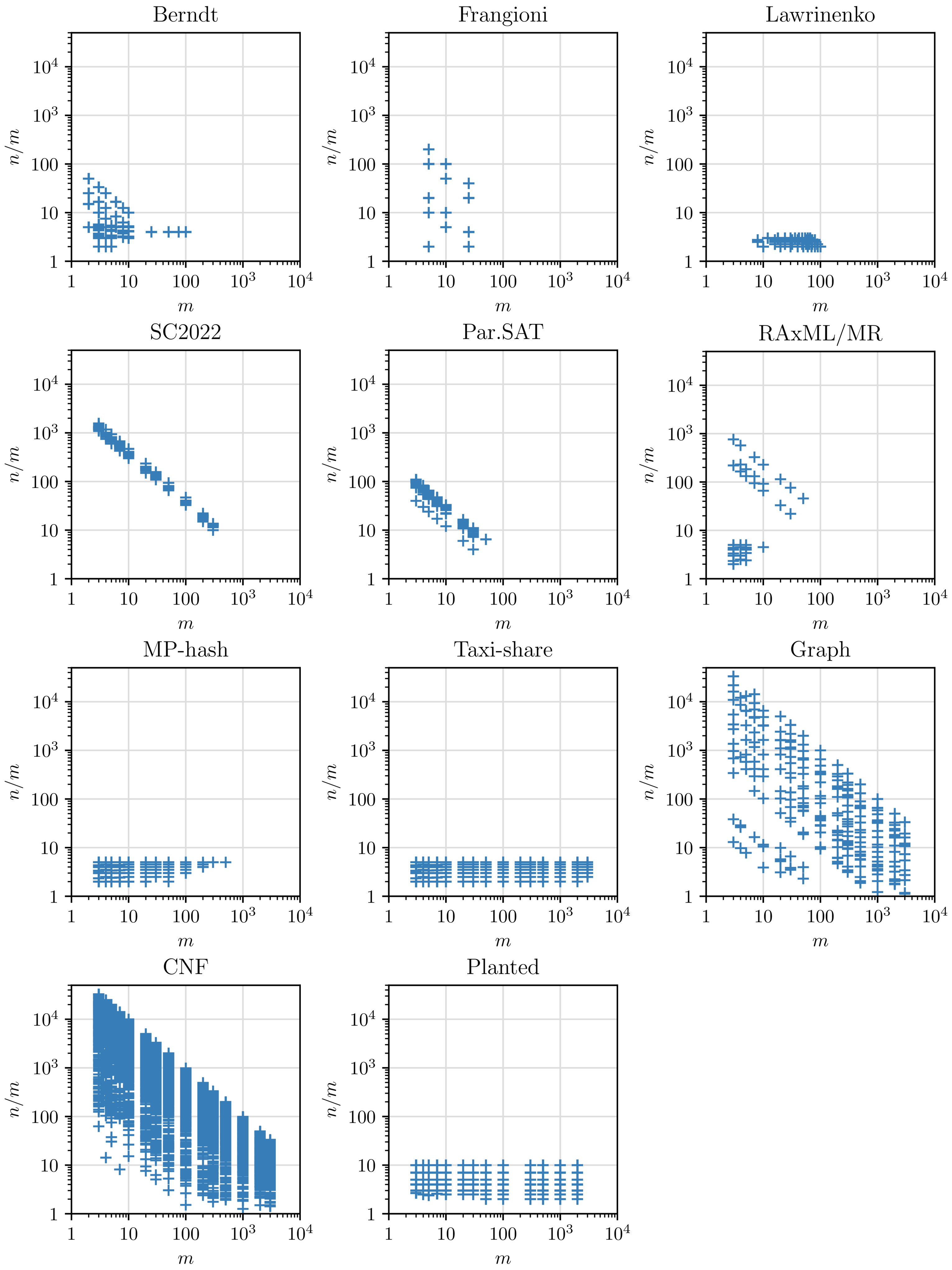}
    \caption{Distribution over values of $m$ and $n/m$ for each benchmark family.}
    \label{fig:benchmark-families-m-distributions}
\end{figure*}

\iffalse
\begin{table}
    \centering
    \footnotesize
    \setlength{\tabcolsep}{2.5pt}
    \begin{tabular}{lrrrrr}
        \toprule
                           &  R5    & R6     & FUR    & Irr.   & CDSM \\
        \midrule
        Solved ratio       &  1.10  & 1.17   & 1.56   & 1.56   & 1.64 \\
        Speedup            &  1.04  & 1.66   & 14.21  & 13.36  & 11.91 \\
        Expl.\ node ratio  &  1.10  & 2.59   & 59.92  & 59.52  & 89.81 \\
        \bottomrule
    \end{tabular}
    \caption{
        Relative improvement over ``Base'' in terms of solved instances, running times, and explored nodes.
        For each ratio, a number $x>1$ denotes an \textit{improvement} over the respective property by factor $x$.
        Speedups and explored node ratios are w.r.t.\ commonly solved instances only.
        Base solved 771 instances at a geom.\ mean running time of 0.045\,s exploring a geom.\ mean of 73778 nodes.
    }
    \label{tab:relative-improvements}
\end{table}
\fi

\begin{figure*}[h!]
    \centering
    \includegraphics[width=0.9\textwidth]{cdfs-by-domain.png}
    \caption{Running times split by family (without Par.SAT since no instances were solved).
            VBS (``Virtual Best Solver'') depicts a virtual running time oracle combining CDSM and ILP.}
    \label{fig:cdfs-by-domain}
\end{figure*}

\fi

\end{document}

% End of ltexpprt.tex 